\def\paperID{2809}
\def\confName{CVPR}
\def\confYear{2025}

\def\paperTitle{Geometry Field Splatting with Gaussian Surfels}

\def\authorBlock{
    Kaiwen Jiang\textsuperscript{1} \quad
    Venkataram Sivaram\textsuperscript{1} \quad
    Cheng Peng\textsuperscript{2} \quad
    Ravi Ramamoorthi\textsuperscript{1} \\
    \textsuperscript{1}University of California, San Diego\quad\textsuperscript{2}Johns Hopkins University
}

\newif\ifreview 
\newif\ifarxiv \newcommand{\arxiv}{\arxivtrue}
\newif\ifcamera 
\newif\ifrebuttal 

\arxiv %

\pdfoutput=1
\documentclass[10pt,twocolumn,letterpaper]{article}
\ifreview \usepackage[review]{cvpr} \fi
\ifarxiv \usepackage[pagenumbers]{cvpr} \fi
\ifrebuttal \usepackage[rebuttal]{cvpr} \fi
\ifcamera \usepackage{cvpr} \fi

\usepackage{graphicx}	
\usepackage{amsmath}	
\usepackage{amssymb}
\usepackage{amsthm}
\usepackage{bm}
\usepackage{bbm}
\usepackage{booktabs}
\usepackage{times}
\usepackage{microtype}
\usepackage{epsfig}
\usepackage[table,xcdraw,dvipsnames]{xcolor}
\usepackage{caption}
\usepackage{float}
\usepackage{placeins}
\usepackage{color, colortbl}
\usepackage{stfloats}
\usepackage{enumitem}
\usepackage{tabularx}
\usepackage{xstring}
\usepackage{multirow}
\usepackage{xspace}
\usepackage{url}
\usepackage{subcaption}
\usepackage{xcolor}
\usepackage[table]{xcolor}
\usepackage[accsupp]{axessibility}  %

\usepackage[hang,flushmargin]{footmisc}

\ifcamera \usepackage[accsupp]{axessibility} \fi

\ifarxiv  \fi

\newcommand{\R}[1]{{%
    \textbf{%
        \ifstrequal{#1}{1}{\textcolor{red}{R#1}}{%
        \ifstrequal{#1}{2}{\textcolor{blue}{R#1}}{%
        \ifstrequal{#1}{3}{\textcolor{magenta}{R#1}}{%
        \ifstrequal{#1}{4}{\textcolor{teal}{R#1}}{%
                           \textcolor{cyan}{R#1}%
        }}}}%
    }%
}}

\makeatletter
\renewcommand{\paragraph}{%
  \@startsection{paragraph}{4}%
  {\z@}{0.2ex \@plus 0.3ex \@minus .2ex}{-1em}%
  {\normalfont\normalsize\bfseries}%
}
\makeatother

\renewcommand{\baselinestretch}{0.98}  %

\usepackage{xr-hyper}

\makeatletter
\newcommand*{\addFileDependency}[1]{
  \typeout{(#1)}
  \@addtofilelist{#1}
  \IfFileExists{#1}{}{\typeout{No file #1.}}
}

\makeatother
\newcommand*{\myexternaldocument}[1]{
    \externaldocument{#1}
    \addFileDependency{#1.tex}
    \addFileDependency{#1.aux}
}

\definecolor{cvprblue}{rgb}{0.21,0.49,0.74}
\usepackage[pagebackref,breaklinks,colorlinks,citecolor=cvprblue]{hyperref}
\usepackage[capitalize]{cleveref}
\crefname{section}{Sec.}{Secs.}
\crefname{table}{Table}{Tables}
\crefname{figure}{Fig.}{Figs.}
\crefname{equation}{Eqn.}{Eqns.}
\crefname{theorem}{Theorem.}{Theorems.}
\crefname{lemma}{Lemma.}{Lemma.}

\newtheorem{theorem}{Theorem}
\newtheorem{lemma}{Lemma}
\newtheorem{property}{Property}

\ifarxiv \crefname{appendix}{App.}{Apps.}
\else \crefname{appendix}{Suppl.}{Suppls.} \fi

\unless\ifarxiv \myexternaldocument{_supplementary} \fi

\begin{document}
\title{\paperTitle}
\author{\authorBlock}

\maketitle

\begin{abstract}
Geometric reconstruction of opaque surfaces from images is a longstanding challenge in computer vision, with renewed interest from volumetric view synthesis algorithms using radiance fields.  %
We leverage the geometry field proposed in recent work for stochastic opaque surfaces, which can then be converted to volume densities.  We adapt Gaussian kernels or surfels to splat the geometry field rather than the volume, enabling precise reconstruction of opaque solids.  Our first contribution is to derive an efficient and almost exact differentiable rendering algorithm  for geometry fields parameterized by Gaussian surfels, while removing current approximations involving Taylor series and no self-attenuation. 
Next, we address the discontinuous loss landscape when surfels cluster near geometry, showing how to guarantee that the rendered color is a continuous function of the colors of the kernels, irrespective of ordering. 
Finally, we use latent representations with spherical harmonics encoded reflection vectors rather than spherical harmonics encoded colors to better address specular surfaces.  
We demonstrate significant improvement in the quality of reconstructed 3D surfaces on widely-used datasets.

\end{abstract}
\section{Introduction}
\label{sec:intro}

{The reconstruction of opaque surfaces from a collection of calibrated RGB images is a classic challenge in computer vision. This problem has received broader attention in recent years due to the success of Neural Radiance Field (NeRF)~\cite{nerf2020} and other volumetric novel view synthesis algorithms. Although methods like NeRF \cite{nerf2020} are effective for view synthesis, they fall short in geometric reconstruction applications, since they do not explicitly reconstruct an accurate surface. Follow up methods, which are specialized towards surface reconstruction, have proposed instead that an explicit geometry representation should parameterize the volume density field.
}

{The typical approach taken by previous methods \cite{idr, volsdf, neus, neus2, neuralangelo, neudf, neuraludf, object_as_volumes} is to parameterize signed or unsigned distance fields with a neural network, which can then be converted to a density field for volume rendering \cite{neus, object_as_volumes}. Eikonal regularization \cite{eik_loss} is often used to promote smoothness of the reconstructed surfaces. In particular, \citet{object_as_volumes} provides a theoretical bridge between surface representations and volume rendering with stochastic geometry, and generalizes the use of distance fields to introduce a general stochastic geometry field, which we refer to as simply a \emph{geometry field}. \citet{dipole} shows that such geometry fields can be parameterized with a %
Poisson field.}

\begin{figure}[tp]
    \centering
    \includegraphics[width=\linewidth]{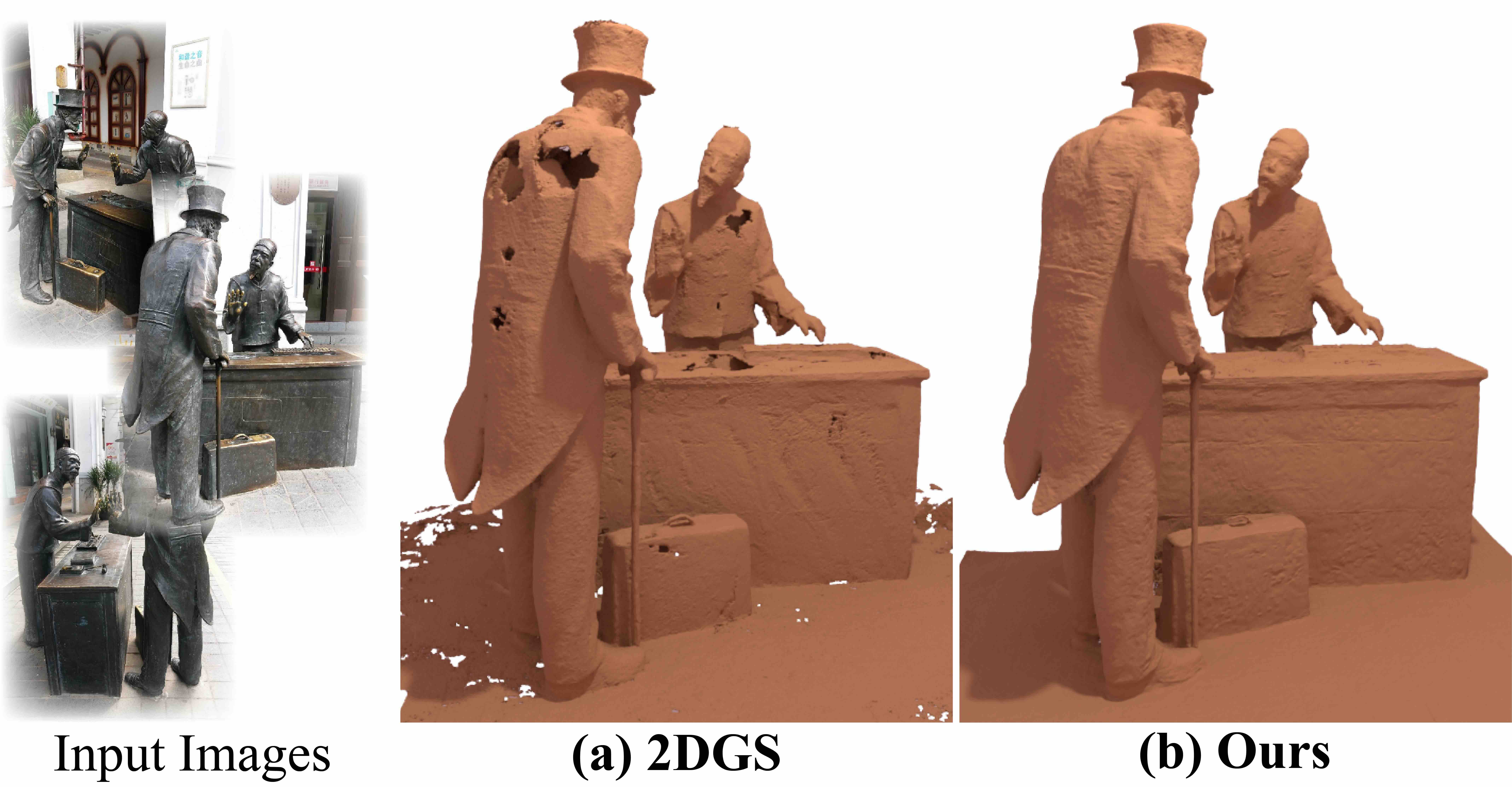}
    \caption{
    We introduce a geometry reconstruction method from a set of calibrated RGB images by splatting a geometry field with Gaussian surfels.
    We enable almost exact and efficient differentiable rendering. Compared to 2DGS, we reach better geometry quality and achieve overall smooth and detailed geometry without having cracks or holes.
    The reader may wish to zoom into the electronic version in the figures. (Note: object colors used are for visualization, not original object color).  
    }
    \label{fig:teaser}
\end{figure}

Recent advances in Gaussian splatting based methods \cite{3dgs} have motivated a new class of surface modeling methods (e.g., \cite{2dgs, radegs, gaussian_surfel, GOF, guedon2023sugar, wolf2024gsmesh}). These methods utilize kernels, i.e., 3D Gaussians or 2D Gaussians, to parameterize a density or opacity field and then perform volume splatting \cite{3dgs, volume_splatting} for rendering. Geometry is then recovered from the converged density or opacity field through varying definitions of the surface. 
The benefit compared to the neural methods comes from the speed and ability to generalize to complex scenes (e.g., \cite{surfel-human-1, surfel-human-2, zhao2024surfel}). Despite the success achieved by these approaches, the surface is not theoretically clearly defined for a general density field, which incurs approximation that limits the reconstruction quality.

In this paper, we {adapt a modern volumetric representation based on Gaussian kernels to splat the geometry field that provides precise geometry for reconstructing opaque solids. We }show that {we can} utilize the kernels, i.e., 2D Gaussians or Gaussian surfels, to parameterize the geometry field and still use the efficient volume splatting algorithm for differentiable rendering. We carefully remove the approximations in the original volume splatting algorithm to ensure almost exact rendering, except for the use of global sorting to approximate per-ray sorting. We also observe that, to make the represented stochastic geometry gradually converge to the deterministic case, kernels need to be clustered around the surface. Such clustering leads to an unstable ordering of kernels in volume splatting and the rendered color a discontinuous function of the properties of kernels. The view synthesis loss then creates a discontinuous loss landscape that is detrimental for optimization. We propose a solution to ensure that the rendered color is a continuous function of the properties of the kernels.  %
Finally, inspired by \cite{gaussianshader, voxurf, refnerf, wang2025unisdf, verbin2024nerf, ge2023ref}, we also explore different choices of color representation other than spherical harmonics \cite{sh} to tackle specular surfaces. 
Through extensive experiments, we show that such a design with geometry clearly defined can significantly improve the quality of the reconstructed surface (see Figs.~\ref{fig:teaser} and~\ref{fig:qualitative}).

{Our main contributions are summarized as follows:
\begin{itemize}
    \item We derive an efficient and differentiable rendering algorithm for geometry fields parameterized by Gaussian surfels (Sec.~\ref{section-method-geo-field-splatting}). Besides, we carefully analyze and reduce the approximations in the rendering algorithm.
    \item We expose and remedy the discontinuity of the loss landscape, which ultimately helps our method to converge from stochastic geometry to a deterministic shape (Sec.~\ref{section-method-continuous-loss}).
\end{itemize}}
\section{Related Work}
\label{sec:related}

\paragraph{Multi-view Stereo Surface Reconstruction.} The computer vision community has studied 3D reconstruction from multi-view images for decades using multi-view stereo \cite{multi-view-stereo, multi-view-stereo-2, 4359315, 5226635, 5989831, 1206509, Bleyer2011PatchMatchS}. Despite their high accuracy, these methods do not provide complete geometry. Therefore, CNN-based approaches (e.g., \cite{laga2020survey, zhu2021deep, WANG2021102102, yao2018mvsnet, yu2020fast, gu2020cascade, ding2022transmvsnet, kd-mvs, cl-mvsnet, patchmatchnet,effimvs,riav}) were introduced to overcome this limitation%
{, but they cannot reconstruct reflective surfaces}. %
Neural surface reconstruction approaches are {later} introduced to {address more general surface types.} %

\paragraph{Neural Surface Reconstruction.} Thanks to the advance of NeRF \cite{nerf2020}, some works (e.g., \cite{idr, neus, neus2, long2022sparseneus, volsdf, neuralangelo, object_as_volumes, ge2023ref, wang2025unisdf}) propose to utilize various representations, such as an MLP \cite{nerf2020} or Hash-grid \cite{instant-ngp}, to parameterize a SDF field and then convert it into a density field for rendering.  The mesh is then extracted using marching cubes from the SDF field after convergence. Identifying the limitation of the SDF field for open-surface objects, some other works (e.g., \cite{neudf, neuraludf}) propose to leverage an UDF field instead of a SDF field. Inspired by multi-view stereo (MVS), another group of methods (e.g., \cite{neuralwarp, geoneus, mvsneus}) add a multi-view consistency loss to further enhance the quality of reconstructed geometry.

\paragraph{Gaussian Splatting based Surface Reconstruction.} \citet{3dgs} propose to utilize 3D Gaussians to parameterize the density field for differentiable view synthesis. Owing to its efficiency and high quality, it has become popular to extend Gaussian splatting (GS) for geometry reconstruction. SuGaR \cite{guedon2023sugar} constrains the 3D Gaussians to be flattened for extracting the geometry but still exhibits holes. 
2DGS \cite{2dgs} and \citet{gaussian_surfel} then propose to directly use the 2D Gaussians to parameterize the density field and utilize depth-normal consistency loss or a monocular normal prior to ensure the smoothness of extracted geometry. 
In contrast, RaDe-GS \cite{radegs} and GOF \cite{GOF} still rely on 3D Gaussians, but propose custom surface definitions for the 3D Gaussian. %
Surfaces are extracted through TSDF fusion \cite{10.1145/3596711.3596726, open3d} or tetrahedral grids \cite{GOF}. 
Different from these methods, we parameterize the geometry field with 2D Gaussians and then convert it into the density field for rendering with refined volume splatting. We still use TSDF fusion to extract the geometry and achieve significant improvements in terms of geometry quality. Concurrently, PGSR \cite{pgsr} introduces a multi-view constraint into the optimization.

\begin{figure}[tp]
    \centering
    \includegraphics[width=\linewidth]{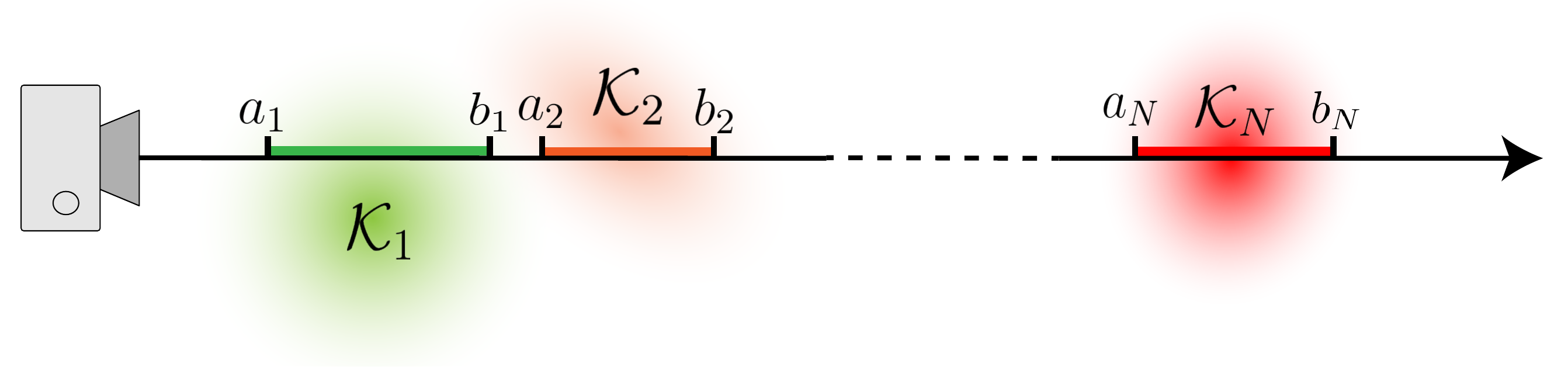}
    \caption{Illustration of a ray intersecting with non-overlapping kernels $\mathcal{K}_1, \mathcal{K}_2, ..., \mathcal{K}_N$, which are sorted based on their intersection intervals $[a_1, b_1], [a_2, b_2], ..., [a_N, b_N]$.}
    \label{fig:vs}
\end{figure}

\section{Preliminaries}
\label{section-method-preliminary}
We first introduce and analyze the volume splatting algorithm and volume representation of opaque solids.
\paragraph{Volume Splatting.} 
Consider a ray $\mathbf{x}(t)=\mathbf{o}+t\boldsymbol{\omega}$, where $\mathbf{o}$ denotes the camera origin, $\boldsymbol{\omega}$ denotes the ray direction and $t > 0$ denotes the depth. Volume splatting \cite{volume_splatting} aims to evaluate the following volume rendering equation:
\begin{equation}
\label{equation-volume-rendering}
    \mathbf{C} = \int_{0}^{\infty} \mathbf{c}(t) \sigma(t) \exp\left(-\int_{0}^{t}\sigma(t')\,dt'\right) dt, 
\end{equation}
where $\mathbf{c}(t), \sigma(t)$ denote the color and density values{, or the extinction coefficent of the volume,} at $\mathbf{x}(t)$.

As shown in \cref{fig:vs}, {\citet{volume_splatting} propose to decompose the density field}
as a set of independent kernels, each of which is chosen to have finite support range by cutoff, i.e., $\sigma(t)=\sum_{i=1}^{N} w_i \mathcal{K}_i(t)$, where $N$ denotes the number of kernels, and $w_i$ and $\mathcal{K}_i(t)$ respectively denote the weight and evaluated value at $t$ for the $i^\text{th}$ kernel. The intersections between these kernels and the ray are assumed to be non-overlapping and, therefore, these kernels are sorted along the ray based on the intersection interval. {The intersection interval of the $i^\text{th}$ kernel is denoted as $[a_i, b_i]$, i.e., $\mathcal{K}_i(t)\neq0$ if and only if $a_i\leq t\leq b_i$.} The color values inside the $i^\text{th}$ kernel are defined to be constant as $\mathbf{c}_i$. 

By ignoring the self-occlusion and expanding the transmittance term with the Taylor Series, \Cref{equation-volume-rendering} can be rewritten as:
\begin{equation}
\label{equation-volume-splatting}
    \mathbf{C} = \sum_{i=1}^N \mathbf{c}_i \rho_i \prod_{j=1}^{i-1} (1 - \rho_j), 
\end{equation}
where $\rho_i$ is the integration of density value along the ray, when considering only the $i^\text{th}$ kernel. Namely, $\rho_i = \int_{a_i}^{b_i} \sigma(t)dt = \int_{a_i}^{b_i} w_i \mathcal{K}_i(t)\, dt$, and it is called the footprint of the $i^\text{th}$ kernel with respect to the current ray. 
Notice that due to the expansion, $\rho_i \in [0, 1]$. %

Therefore, it is important that the footprint function be easy to evaluate efficiently, such that the evaluation of \cref{equation-volume-splatting} can be efficient.
Volume splatting chooses to use the 3D Gaussian as the kernel, and analytically solve the footprint function by approximating the perspective transform.

In summary, volume splatting makes the following assumption and approximation:
\begin{itemize}
    \item Assume there is no overlapping of any two intersections between kernels and the ray.
    \item Self-occlusion of each kernel is ignored, and transmittance terms and footprint functions are approximated.
\end{itemize}

Even though remarkable view synthesis from a set of calibrated images is achieved by making the volume splatting differentiable \cite{3dgs}, the non-overlapping assumption leads to a discontinuous view synthesis loss function and the constraint that the footprint function is within the range of $[0, 1]$. We will show that these problems are detrimental for the surface reconstruction and propose solutions. Besides, \citet{3dgs} uses a global sorting to approximate per-ray sorting.

\paragraph{Stochastic Geometry Representation for Opaque Solid.}
{\citet{object_as_volumes} propose that} for an opaque solid existing in the $\mathbb{R}^3$ space, each point $\mathbf{x}\in\mathbb{R}^3$ is associated with a random variable $G(\mathbf{x}) \in \mathbb{R}$. Without loss of generality, the random variable is chosen to obey the normal distribution, i.e., $G(\mathbf{x})\sim\mathcal{N}(\mu(\mathbf{x}), 1/s(\mathbf{x}))$, where $\mu(\mathbf{x})$ and $s(\mathbf{x})$ denote the mean and inverse of standard deviation. Each point $\mathbf{x}$ is then associated with an occupancy value $o(\mathbf{x})$ and vacancy value $v(\mathbf{x})$, defined as:
\begin{equation}
\label{equation-occupancy-vacancy}
\begin{aligned}
    o(\mathbf{x}) = \Pr\{G(\mathbf{x})\geq0\} &= \Psi(\mu(\mathbf{x})s(\mathbf{x})) \\
    v(\mathbf{x}) = \Pr\{G(\mathbf{x})<0\} &= \Psi(-\mu(\mathbf{x})s(\mathbf{x})), 
\end{aligned}
\end{equation}
where $\Psi(\cdot)$ denotes the cdf of the standard normal distribution. The solid is defined at $\mathbbm{1}_{\{G(\mathbf{x})\geq0\}}$.

From \cite{object_as_volumes}, the density at $\mathbf{x}$ is defined as:
\begin{equation}
\label{equation-object-volume-density}
\begin{aligned}
    \sigma(\mathbf{x}) =& \frac{||\nabla v(\mathbf{x})||}{v(\mathbf{x})} \cdot \left(\alpha(\mathbf{x})|\boldsymbol{\omega}\cdot\mathbf{n}(\mathbf{x})| + \frac{1-\alpha(\mathbf{x})}{2}\right) \\
    =& \frac{\psi(-\mu(\mathbf{x})s(\mathbf{x}))}{\Psi(-\mu(\mathbf{x})s(\mathbf{x}))}||\nabla (\mu(\mathbf{x})s(\mathbf{x}))|| \cdot \\
    & \left(\alpha(\mathbf{x})|\boldsymbol{\omega}\cdot\mathbf{n}(\mathbf{x})| + \frac{1-\alpha(\mathbf{x})}{2}\right),
\end{aligned}
\end{equation}
where {$\psi(\cdot)$ denotes the pdf of the standard normal distribution,} $\mathbf{n}(\mathbf{x})= -\nabla(\mu(\mathbf{x})s(\mathbf{x})) / || \nabla (\mu(\mathbf{x})s(\mathbf{x})) ||$ and $\alpha(\mathbf{x}) \in [0, 1]$ denotes an anisotropic parameter which is close to $1$ near the surface, and close to $0$ in the interior of the opaque solid. As $s(\mathbf{x})\to\infty$, the stochastic geometry converges to the deterministic case.

Notice that, in other works (e.g.,  \cite{object_as_volumes, neus}), the $\mu(\mathbf{x})$ is defined to be a SDF field. Defining it as a SDF field allows the usage of Eikonal loss to promote the smoothness, but there is no constraint over the meaning of $\mu(\mathbf{x})$.

Since $\mu(\mathbf{x})$ and $s(\mathbf{x})$ are always together in \cref{equation-object-volume-density}, we define that
\begin{equation}
\label{equation-geometry-field}
    F(\mathbf{x}) := \mu(\mathbf{x})s(\mathbf{x}), 
\end{equation}
which is called the \emph{geometry field} as it indicates the underlying geometry. When $\forall \mathbf{x} \in\mathbb{R}^3, |F(\mathbf{x})|\to\infty$, the represented geometry is assumed to be converged to the deterministic case. \cref{equation-object-volume-density} can be then written as:
\begin{equation}
\label{equation-density-ours-before}
    \sigma(\mathbf{x}) = \frac{\psi(-F(\mathbf{x}))}{\Psi(-F(\mathbf{x}))}||\nabla F(\mathbf{x})|| \cdot \left(\alpha(\mathbf{x})|\boldsymbol{\omega}\cdot\mathbf{n}(\mathbf{x})| + \frac{1-\alpha(\mathbf{x})}{2}\right), 
\end{equation}
where $\mathbf{n}(\mathbf{x})=-\nabla F(\mathbf{x}) / ||\nabla F(\mathbf{x})||$.

\section{Method}
\label{sec:method}

In \cref{section-method-geo-field-splatting}, we describe how we can reduce the approximations taken by the splatting algorithm. We then derive a refined splatting method which enables differentiable rendering of the opaque solid represented by the geometry field. Next, in \cref{section-method-continuous-loss}, we propose a novel solution to address the issue of discontinuous loss landscapes which arises as the stochastic geometry gradually becomes deterministic. Finally, \cref{section-method-color-representation} introduces an improvement in color representation to account for specular surfaces.

\subsection{Geometry Field Splatting}
\label{section-method-geo-field-splatting}
\paragraph{Revisit Splatting Algorithm.}
As the volume splatting is an \emph{approximate} rendering algorithm, we aim to reduce the approximations, and investigate the situation under which the refined volume splatting algorithm is \emph{exact}.

Specifically, we do not ignore the self-occlusion and do not expand the transmittance term with the Taylor series. \cref{equation-volume-splatting} is then replaced with:
\begin{equation}
\label{equation-volume-splatting-ours-first}
    \mathbf{C} = \sum_{i=1}^N \mathbf{c}_i (1-\exp(-\rho_i)) \prod_{j=1}^{i-1} \exp(-\rho_j). 
\end{equation}
Please find the derivation in A.1 of the supplementary document. 
Notice that then $\rho_i \in \mathbb{R}_+$ with no further restrictions.

We are now left with the non-overlapping assumption and footprint function approximation. We notice the following property:
\begin{property}
    When the intersection intervals of two and only two kernels $\mathcal{K}_i, \mathcal{K}_j$ fully overlap{, i.e., $a_i = a_j \land b_i = b_j$}, and they have the same color, it is equivalent to having one kernel $\mathcal{K}_{ij}$ there with $\rho_{ij}=\rho_{i}+\rho_{j}$ and the same color.
\end{property}
\begin{proof}
    As these two kernels are the only two kernels whose intersection intervals fully overlap at $[a_i, b_i]$%
    , we can conclude that $j=i+1$ without losing generality. Therefore, $(1-\exp(-\rho_i)) + (1 - \exp(-\rho_j))\exp(-\rho_i)=1-\exp(-(\rho_i + \rho_j))$. Notice that the alpha-blended value of these two kernels is equivalent to having a single kernel there.
\end{proof}
This property can be extended into the following form through mathematical induction:
\begin{lemma}
\label{lemma-k-overlap}
    When the intersection intervals of $m$ kernels $\mathcal{K}_{i_1}, \mathcal{K}_{i_2}, ..., \mathcal{K}_{i_m}$ fully overlap{, i.e., $a_{i_1} = a_{i_2} = ... = a_{i_m} \land b_{i_1} = b_{i_2} = ... = b_{i_m}$}, and they have the same color, it is equivalent to having one kernel $\mathcal{K}$ there with $\rho=\rho_{i_1}+\rho_{i_2}+...+\rho_{i_m}$ and the same color.
\end{lemma}

Therefore, we can define an operator $\bigoplus$, which satisfies the commutative, associative and identity properties, such that $\mathcal{K}_{ij}=w_i\mathcal{K}_i\bigoplus w_j\mathcal{K}_j$ where $\rho_{ij} = \rho_{i} + \rho_{j}$. When the kernels on the ray have either non-overlapping or fully overlapped intersection intervals, the refined splatting algorithm defined by \cref{equation-volume-splatting-ours-first} is an \emph{exact} rendering algorithm of density field defined as: $\sigma(\mathbf{x})={\bigoplus_{i=1}^N}w_i\mathcal{K}_i(\mathbf{x})${, because we remove the approximations when evaluating \cref{equation-volume-rendering}}.

By considering the remaining approximation, we can then reach the following theorem:
\begin{theorem}
\label{theorem:unbiased}
    When 1) the intersection intervals of any two kernels are non-overlapping, or fully overlapped in which case these two kernels have the exact same color value, and, 2) the footprint function is accurately calculated, then the sorting of kernels is well-defined and \cref{equation-volume-splatting-ours-first} gives an exact rendering of a density field which is defined as: $\sigma(\mathbf{x})={\bigoplus_{i=1}^N}w_i\mathcal{K}_i(\mathbf{x})$.
\end{theorem}

\paragraph{Parameterize and Render Geometry Field with Gaussian Surfels.}
Our overall idea is illustrated in \cref{fig:pipeline}. We propose to parameterize the geometry field in \cref{equation-geometry-field} with kernels, then demonstrate how to convert it into the density field for volume splatting, and reveal under which conditions it is an exact rendering algorithm. We will also show that choosing Gaussian surfels as our kernels helps build a both efficient and almost exact rendering algorithm.

Specifically, we decompose the $F(\mathbf{x})$ as:
\begin{equation}
\label{equation-our-geometry-field-definition}
    F(\mathbf{x}) = {\bigoplus_{i=1}^N} w_i\mathcal{K}_i(\mathbf{x}) - c, 
\end{equation}
where $\bigoplus$ is the custom plus operator which we have just discussed and will elaborate on later, $\omega_i>0$, and $c$ is a manually chosen constant. Since we expect the kernels to cluster around the surface, the anistropic parameter $\alpha(\mathbf{x})$ is then set to $1$ as discussed in \cite{object_as_volumes}.

We first discuss the property of such a formulation of $F(\mathbf{x})$, when the represented geometry converges to the deterministic case. At the place where it is vacant, i.e., there are no kernels, $F(\mathbf{x}) = -c$, such that $v(\mathbf{x})\approx 1$. Therefore, $c$ is a manually chosen large positive value. At the place where kernels cluster, $F(\mathbf{x})$ is expected to be a large positive value, such that $o(\mathbf{x})\approx 1$. Therefore, when a ray travels in the space, the transmittance is expected to almost immediately fall from $1$ to $0$ at the place where it first intersects with kernels as shown in \cref{fig:pipeline} (c).

We still start from the assumption that there is no overlapping of any two intersections between kernels and the ray and then remove this assumption. From \cref{equation-density-ours-before}, the corresponding density field is defined as:
\begin{equation}
\label{equation-density-ours}
    \sigma(\mathbf{x}) = \frac{\psi(-F(\mathbf{x}))}{\Psi(-F(\mathbf{x}))} ||\nabla F(\mathbf{x})||\cdot|\boldsymbol{\omega}\cdot\mathbf{n}(\mathbf{x})|.
\end{equation}

\begin{figure}[tp]
    \centering
    \includegraphics[width=\linewidth]{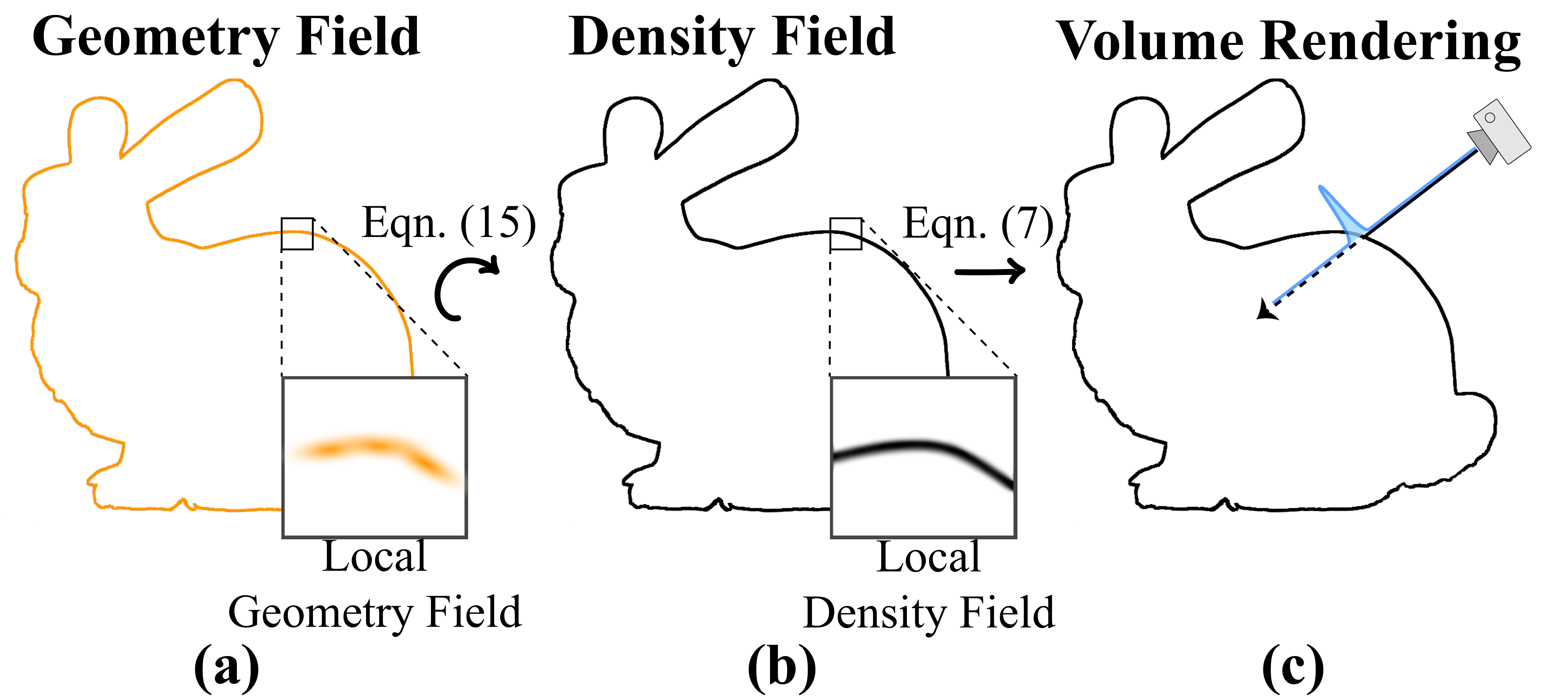}
    \caption{Overview of our algorithm. {\textbf{(a)} We first use 2D Gaussians to parameterize the geometry field, $F$. These kernels are expected to cluster around the surface. \textbf{(b)} We then convert the geometry field into the density field $\sigma$ and lastly, \textbf{(c)} we leverage our refined volume splatting algorithm for differentiable rendering.}}
    \label{fig:pipeline}
\end{figure}

As \cref{equation-volume-splatting-ours-first} indicates, given a ray, our goal is to evaluate the footprint function for the $i^\text{th}$ kernel:
\begin{equation}
\label{equation-rho-ours}
    \rho_i = \int_{a_i}^{b_i} \sigma(\mathbf{x}(t)) dt.
\end{equation}
We now make our choice of kernels. Aiming at exact rendering, we choose to use 2D Gaussians or Gaussian surfels \cite{2dgs, gaussian_surfel, Surfels} as our kernels, instead of 3D Gaussians \cite{3dgs}. In general, the intersections between Gaussian surfels and a ray are points. Thus, any two of the intersections are either not overlapping at all or fully coincide, satisfying the conditions of \cref{theorem:unbiased}. 

\begin{figure}[tp]
    \centering
    \includegraphics[width=0.9\linewidth]{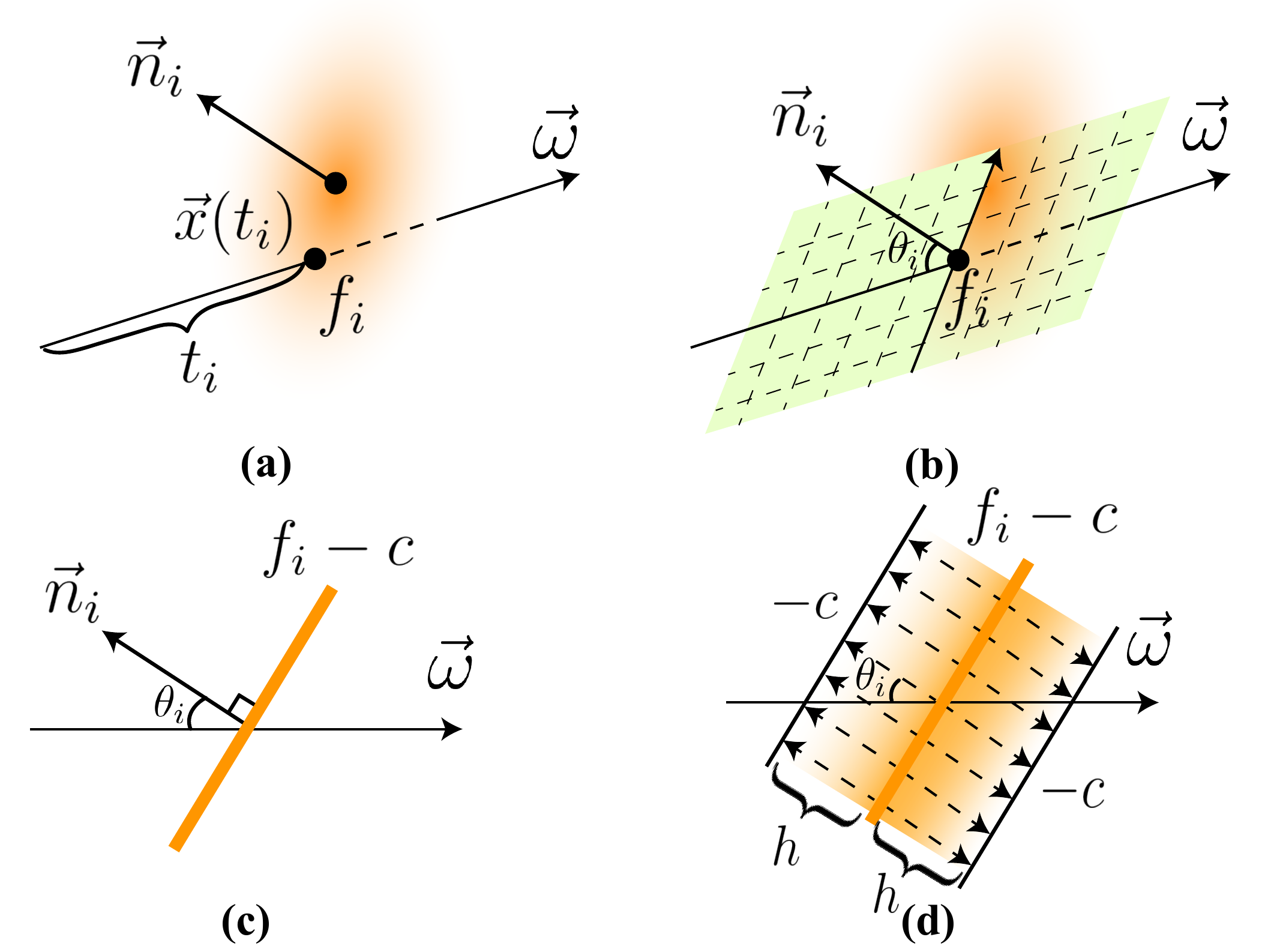}
    \caption{
    \textbf{(a)} An illustration of intersecting the $i^\text{th}$ Gaussian surfel, whose normal vector is denoted as $\mathbf{n}_i$ {that is perpendicular to the plane containing the surfel}, with a ray, whose direction is denoted as $\bm \omega$. The depth of the intersected point is denoted as $t_i$, and the intersected point is denoted as $\mathbf{x}(t_i)$. The value of the Gaussian surfel at $\mathbf{x}(t_i)$ is denoted as $f_i$.
    \textbf{(b)} We analyze the intersection by creating a 2D coordinate plane which is parallel to the $\mathbf{n}_i$ and $\mathbf{\omega}$, and passes through $\mathbf{x}(t_i)$. The angle between $\mathbf{n}_i$ and $\bm \omega$ is denoted as $\theta_i$. 
    \textbf{(c)} On the 2D coordinate plane, within an infinitesimally small range, the intersected line{, drawn as an orange line,} between the Gaussian surfel and the plane can be seen as having constant geometry field value $F = f_i-c$.
    \textbf{(d)} We expand the intersected line by linearly decaying it into $-c$ with length $h$ and direction $\mathbf{n}_i$.{ It gives the surfel a 3D width, which does not follow the Gaussian distribution. }As $h\to0$, it can be seen as equivalent to (c).
    }
    \label{fig:density-deduction}
\end{figure}

The $i^\text{th}$ 2D Gaussian kernel $\mathcal{G}(\mathbf{x})$ is defined as:
\begin{equation}
    \mathcal{G}_i(\mathbf{x}) = \exp\left(-\frac{1}{2}(\mathbf{x} - \mathbf{m}_i)^T \Sigma_i^{-1} (\mathbf{x} - \mathbf{m}_i)\right), 
\end{equation}
where $\mathbf{m}_i$ denotes its center and $\Sigma_i$ denotes its covariance matrix. As shown in \cref{fig:density-deduction} (a), given a ray $\mathbf{x}(t)=\mathbf{o}+t\boldsymbol{\omega}$, we are able to calculate the intersected depth $t_i$ and evaluate its kernel value at $\mathbf{x}(t_i)$ as $f_i=w_i\mathcal{G}_i(\mathbf{x}(t_i))$ \cite{2dgs}. We then have $F(\mathbf{x}(t))$ near $t_i$:
\begin{equation}
    F(\mathbf{x}(t)) = \begin{cases} 
      f_i-c & t = t_i \\
      -c & |t - t_i|<\epsilon, \text{where }\epsilon\to0.
   \end{cases}
\end{equation}
Such a function is not even a continuous function.
It seems that we cannot calculate \cref{equation-density-ours}. However, since we are interested in the footprint function instead of the density field, we can reach a closed form solution by slightly modifying this function, which makes the splatting algorithm a perfect fit. Specifically, we start analyzing the intersection by creating a plane as a 2D coordinate system in \cref{fig:density-deduction} (b). We look at an infinitesimally small range around $\mathbf{x}(t_i)$ as shown in \cref{fig:density-deduction} (c), such that the geometry field at the intersected line, drawn as an orange line, between the Gaussian surfel and the plane is constant, i.e., $f_i-c$. 
We enable the calculation of \cref{equation-density-ours} and \cref{equation-rho-ours} by {extruding} the intersected line %
with a linear decay. The geometry field value decays linearly from $f_i-c$ to $-c$ with the direction of $\mathbf{n}_i${, which is perpendicular to the plane containing the surfel,} and length of $h$, as shown in \cref{fig:density-deduction} (d). It is equivalent to the original case as $h\to0$. The $F(\mathbf{x}(t))$ near $t_i$ is modified as:
\begin{figure}[tp]
    \centering
    \includegraphics[width=0.9\linewidth]{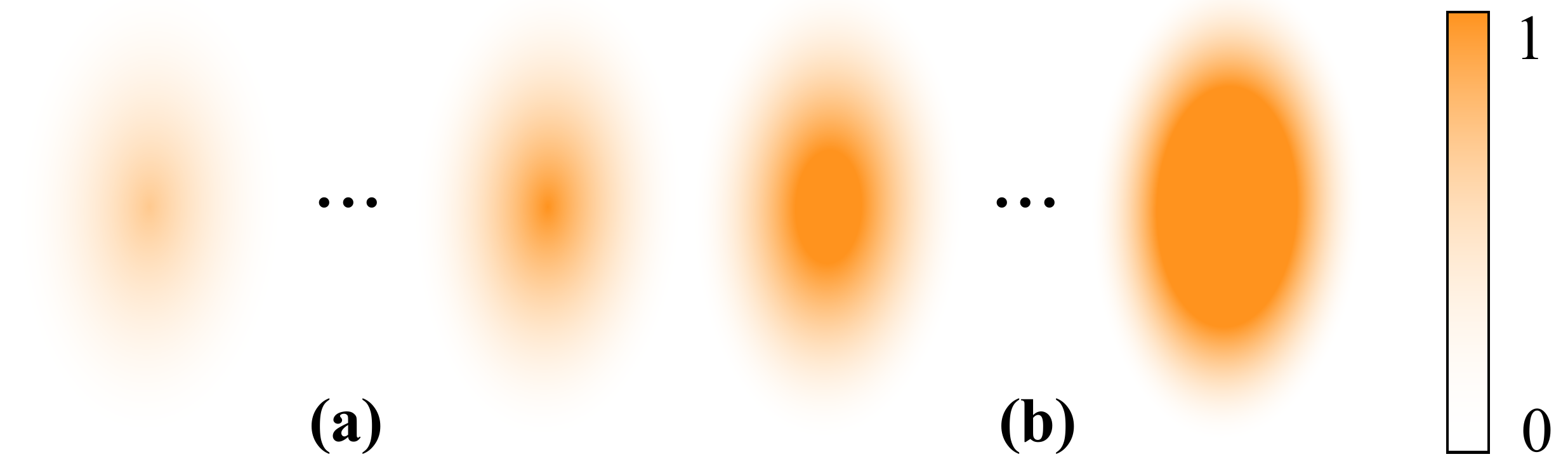}
    \caption{Illustration of opacity value on different Gaussian surfels. \textbf{(a)} 2DGS only allows the opacity to reach $1$ at the center. \textbf{(b)} Due to the enforced transformation from geometry field into density as in \cref{equation-rho-result-ours}, the opacity on a Gaussian surfel in our case does not follow a Gaussian distribution{, thus differentiating it from Gaussian splatting}. Furthermore, we also allow a larger central area to reach $1$, thus making the Gaussian surfel more opaque which benefits the surface reconstruction.}
    \label{fig:primitive_warehouse}
\end{figure}
\begin{equation}
    F(\mathbf{x}(t)) = f_i \times \left(1-\frac{|t-t_i|}{h/\cos\theta_i}\right) -c, \text{when }|t - t_i| < h/\cos\theta_i, 
\end{equation}
where $\cos\theta_i=|\boldsymbol{\omega}\cdot\mathbf{n}_i|$. The $||\nabla F(\mathbf{x})||$ near $t_i$ is then calculated as:
\begin{equation}
    ||\nabla F(\mathbf{x}(t)) || = f_i / h, \text{when }|t - t_i| < h/\cos\theta_i
\end{equation}
and it is parallel to $\mathbf{n}_i$. 
Therefore, we can evaluate \cref{equation-rho-ours} with \cref{equation-density-ours} as:
\begin{equation}
\label{equation-rho-result-ours}
\begin{aligned}
    \rho_i &= \int_{t_i - h/\cos\theta_i}^{t_i + h/\cos\theta_i} \frac{\psi(-F(\mathbf{x}(t)))}{\Psi(-F(\mathbf{x}(t)))} ||\nabla F(\mathbf{x}(t))||\cdot|\boldsymbol{\omega}\cdot\mathbf{n}(\mathbf{x}(t))| dt \\
    &= -2\ln\Psi(c-f_i). 
\end{aligned}
\end{equation}
Please find the derivation in A.2 of the supplementary. 
We define $g(u) := -2\ln\Psi(c-u)$ {to simplify the notation}. %
{As the expansion goes to $0$, i.e., without extrusion,} $\underset{{h\to0}}{\lim}\rho_i=g(f_i)$. 
Notice that the value range of $g(u)$ exceeds $[0, 1]$, making approximated volume splatting in \cref{equation-volume-splatting} not applicable and our refined volume splatting in \cref{equation-volume-splatting-ours-first} suitable. 
We can then use it for rendering in \cref{equation-volume-splatting-ours-first} as:
\begin{equation}
\label{equation-color-ours-final}
    \mathbf{C} = \sum_{i=1}^N \mathbf{c}_i (1-\exp(-g(f_i))) \prod_{j=1}^{i-1} \exp(-g(f_j)). 
\end{equation}

We now remove the non-overlapping assumption. 
We define $S(u) := -2\ln\Psi(c - u)$. Notice that $S(0) = 0$, and $S(u)$ is {a }monotonically increasing {function} {with respect to $u$}. Therefore, there exists a $S^{-1}(v), \forall v \geq 0$. We then define $\bigoplus(\cdot)$, such that $a\bigoplus b=S^{-1}(S(a)+S(b))$. Besides, when $m$ Gaussian surfels $\mathcal{G}_{j+1}, \mathcal{G}_{j+2}, ..., \mathcal{G}_{j+m}, j\in\mathbb{N}_+, m\in\mathbb{N}_+$ intersect the ray at the same place $\mathbf{x}(t_{j+1})$, from \cref{lemma-k-overlap}, \cref{equation-volume-splatting-ours-first} is equivalent to having a kernel whose footprint function equals to $\sum_{k=1}^{m}g(w_{j+k}\mathcal{G}_{j+k}(\mathbf{x}(t_{j+1})))$. 

Given that $$g({\bigoplus_{k=1}^{m}}w_{j+k}\mathcal{G}_{j+k}(\mathbf{x}(t_{j+1})))=\sum_{k=1}^{m}g(w_{j+k}\mathcal{G}_{j+k}(\mathbf{x}(t_{j+1}))), $$ it then corresponds to the geometry field value: 
\begin{equation}
\begin{aligned}
    F(\mathbf{x}(t_{j+1}))&={\bigoplus_{k=1}^{m}}w_{j+k}\mathcal{G}_{j+k}(\mathbf{x}(t_{j+1}))-c \\
    &={\bigoplus_{i=1}^{N}}w_{i}\mathcal{G}_{i}(\mathbf{x}(t_{j+1}))-c, 
\end{aligned}
\end{equation}
as other kernels are $0$ here. It follows our definition \cref{equation-our-geometry-field-definition}.

\begin{figure}[tp]
    \centering
    \includegraphics[width=\linewidth]{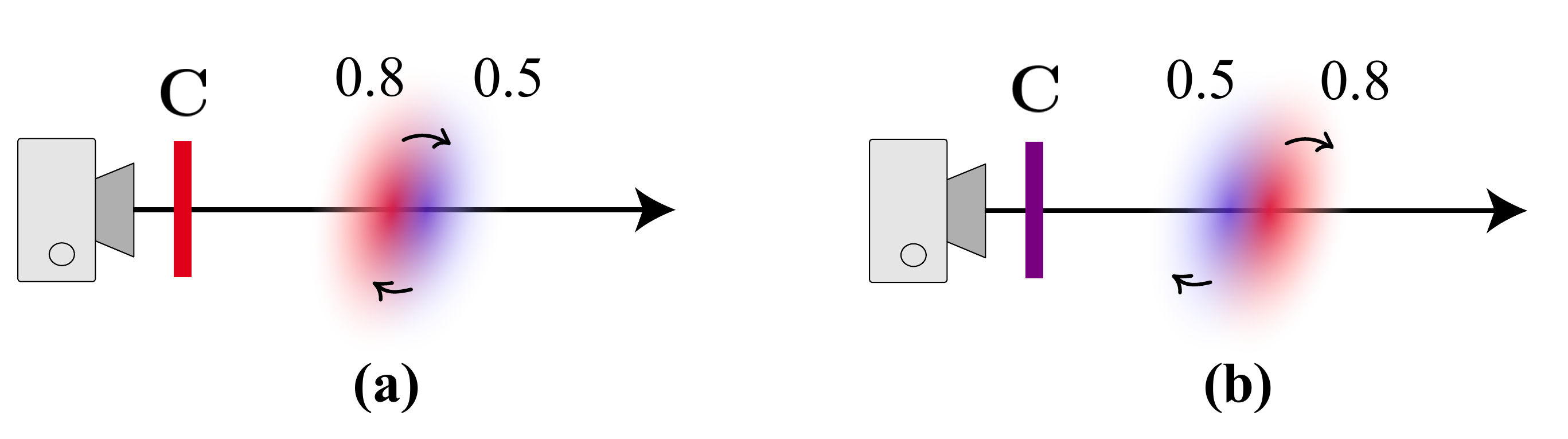}
    \caption{We illustrate the discontinuous change of {rendered} color {in the splatting algorithm} with respect to the positions of kernels. As shown in (a) and (b), when two kernels {continuously change their positions and are eventually swapped}, the ordering could be unstable and the rendered color $\mathbf{C}$ undergoes a discontinuous change when the ordering changes. The opacity is denoted on top of the kernel.}
    \label{fig:order}
\end{figure}

Another benefit that comes with using Gaussian surfels is that we can have the expected depth $D$ efficiently calculated with volume splatting as:
\begin{equation}
\label{equation-volume-splatting-ours-depth}
    D = \sum_{i=1}^N t_i (1-\exp(-g(f_i))) \prod_{j=1}^{i-1} \exp(-g(f_j)). 
\end{equation}
In contrast, with 3D Gaussians, the assumption for volume splatting that the property, i.e., depth in this case, which we want to render, is constant inside the kernel does not hold.

\paragraph{Discussion.} From \cref{theorem:unbiased}, {if} Gaussian surfels {have} the same color when they have the same intersected point, the rendering is then exact, but we make the only approximation to use the global sorting to approximate the per-ray sorting as in \cite{3dgs, 2dgs} for efficiency. 
We will discuss in \cref{section-method-continuous-loss} that this color constraint is also important for having a continuous loss landscape for optimization.

In practice, as in \cref{fig:primitive_warehouse}, the distributions of opacity on the Gaussian surfels do not follow the Gaussian distribution anymore, thus differentiating our algorithm from Gaussian splatting. We also allow larger areas on the Gaussian surfels to become fully opaque. These benefit the surface reconstruction as {shown} in \cref{sec:experiment}.

\subsection{Remedying Loss Landscape Defects}
\label{section-method-continuous-loss}
With the rendering algorithm we introduced before, during training, we apply the view synthesis loss as $\mathcal{L}_\text{rgb}$, depth distortion loss as $\mathcal{L}_d$ and depth-normal consistency loss as $\mathcal{L}_n$, following \cite{2dgs}. However, we use the expected depth defined in \cref{equation-volume-splatting-ours-depth} instead of the median depth in \cite{2dgs} to calculate $\mathcal{L}_n$. The final loss is defined as:
\begin{equation}
    \mathcal{L} = \mathcal{L}_\text{rgb} + \lambda_1 \mathcal{L}_d + \lambda_2 \mathcal{L}_n, 
\end{equation}
where $\lambda_1$ is set specific to the dataset, and $\lambda_2 = 0.05$.
$\mathcal{L}_d$ helps encourage the Gaussian surfels to cluster together, such that the represented stochastic geometry becomes deterministic, $\mathcal{L}_d$ helps smooth the geometry, and $\mathcal{L}_\text{rgb}$ is the main loss for driving the surface reconstruction. However, $\mathcal{L}_\text{rgb}$ is not a continuous function of properties, including center and covariance, of the Gaussian surfel. As shown in \cref{fig:order}, when two Gaussian surfels continuously change their intersected points and swap their position, the rendered color undergoes a discontinuous change. In the case of surface reconstruction, Gaussian surfels are encouraged to cluster together, which makes the ordering unstable and optimization harder.

\begin{figure}[tp]
    \centering
    \includegraphics[width=\linewidth]{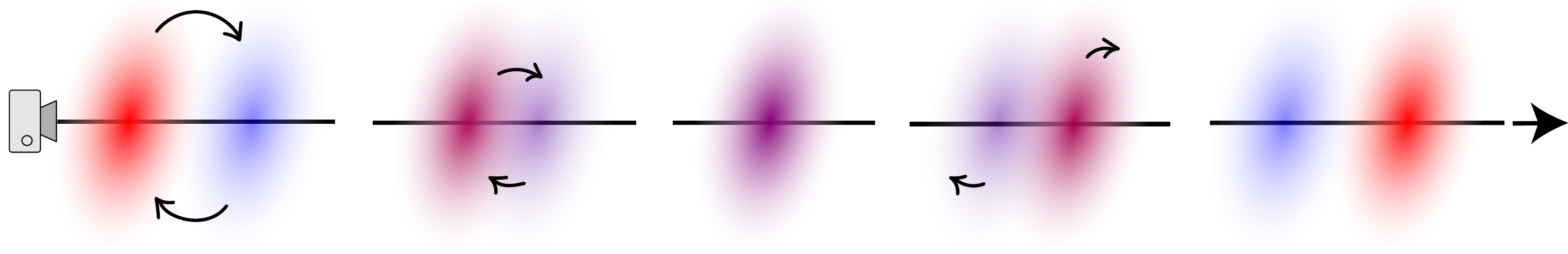}
    \caption{Demonstration of {continuously} changing the colors {of kernels} when two kernels gradually swap their positions.}
    \label{fig:color_merge}
\end{figure}

\begin{figure*}[tp]
    \centering
    \includegraphics[width=1.04\linewidth]{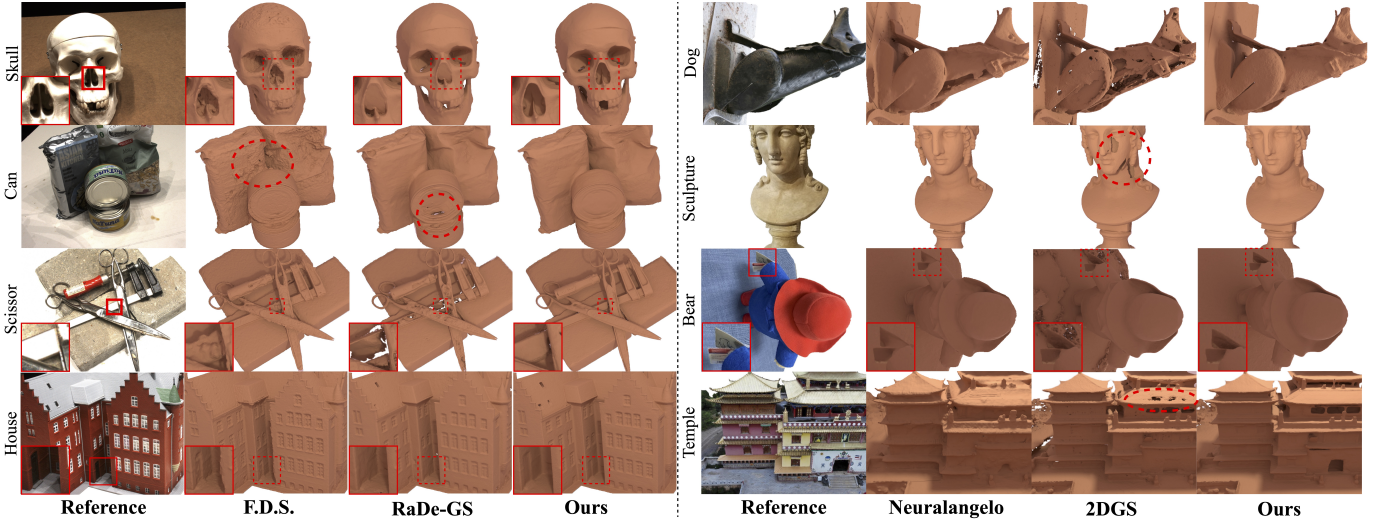}
    \caption{Qualitative comparison on DTU and BlendedMVS datasets. Non-obvious differences are highlighted by insets and red circles. ``F.D.S.'' denotes \citet{dipole}.}
    \label{fig:qualitative}
\end{figure*}

\begin{table*}[thb]
\centering
\footnotesize
\setlength{\tabcolsep}{5pt}
\begin{tabular}{r|ccccccccccccccc|cc}
\hline
\multicolumn{1}{c|}{\textbf{Methods}} & \textbf{24} & \textbf{37} & \textbf{40} & \textbf{55} & \textbf{63} & \textbf{65} & \textbf{69} & \textbf{83} & \textbf{97} & \textbf{105} & \textbf{106} & \textbf{110} & \textbf{114} & \textbf{118} & \textbf{122} & \textbf{Avg. $\downarrow$} & \textbf{Time} \\ \hline
NeuS2 \cite{neus2} & 
0.59 & 
0.79 & 
0.74 & 
0.40 & 
0.93 & 
0.74 & 
0.85 & 
1.37 & 
1.37 & 
0.82 & 
0.74 & 
0.86 & 
0.44 & 
0.60 & 
0.58 & 
0.79 &
\cellcolor{orange!25}$4$m
\\
N.A. \cite{neuralangelo} & 
\cellcolor{orange!25}0.39 & 
0.79 & 
0.38 & 
\cellcolor{red!25}0.30 & 
\cellcolor{yellow!25}0.82 & 
0.72 & 
1.91 & 
1.20 & 
1.94 & 
0.67 & 
\cellcolor{red!25}0.38 & 
1.15 & 
\cellcolor{orange!25}0.34 & 
0.84 & 
0.95 & 
0.85
& $>$ 12h \\
O.V. \cite{object_as_volumes} & 
1.53 & 
1.21 & 
1.06 & 
0.74 & 
\cellcolor{red!25}0.72 & 
1.14 & 
0.75 & 
\cellcolor{orange!25}0.98 & 
1.47 & 
0.86 & 
0.69 & 
\cellcolor{yellow!25}0.80 & 
0.56 & 
0.62 & 
0.73 & 
0.92 & 
$>$ 6h  \\
NeUDF \cite{neudf} & 
2.01 & 
0.75 & 
1.35 & 
0.57 & 
0.95 & 
0.75 & 
0.89 & 
\cellcolor{red!25}0.91 & 
1.29 & 
\cellcolor{red!25}0.56 & 
0.86 & 
1.41 & 
0.58 & 
0.77 & 
0.59 & 
0.95 & 
$>$ 6h \\
F.D.S. \cite{dipole} & 
0.44 & 
\cellcolor{yellow!25}0.66 & 
\cellcolor{orange!25}0.33 & 
\cellcolor{orange!25}0.32 & 
0.90 & 
\cellcolor{yellow!25}0.67 & 
\cellcolor{red!25}0.54 & 
1.28 & 
\cellcolor{orange!25}1.10 & 
0.70 & 
0.65 & 
\cellcolor{red!25}0.57 & 
0.41 & 
\cellcolor{red!25}0.40 & 
\cellcolor{orange!25}0.39 & 
\cellcolor{orange!25}0.62 & 
1h
\\
\hline
SuGaR \cite{guedon2023sugar} & 
1.47 & 1.33 & 1.13 & 0.61 & 2.25 & 1.71 & 1.15 & 1.63 & 1.62 & 1.07 & 0.79 & 2.45 & 0.98 & 0.88 & 0.79 & 1.33 & $42$m
\\
2DGS \cite{2dgs} & 
0.41 & 
0.67 & 
\cellcolor{orange!25}0.33 & 
0.36 & 
0.95 & 
0.81 & 
0.77 & 
1.24 & 
1.19 & 
0.68 & 
0.65 & 
1.28 & 
\cellcolor{yellow!25}0.35 & 
0.65 & 
0.46 & 
0.72 & 
\cellcolor{yellow!25}$6$m
\\ 
G.S. \cite{gaussian_surfel} &
0.64 &
0.85 &
0.58 &
0.43 &
0.99 &
1.18 &
0.90 &
1.14 &
\cellcolor{yellow!25}1.11 &
0.81 &
0.92 &
1.46 &
0.53 &
0.61 &
0.58 &
0.85 & 
\cellcolor{red!25}$3$m
\\
GOF \cite{GOF} &
0.46 &
0.68 &
0.39 &
0.38 &
1.15 &
0.82 &
0.70 &
1.14 &
1.21 &
0.65 &
0.68 &
1.04 &
0.48 &
0.68 &
0.49 &
0.73 &
$33$m
\\
RaDe-GS \cite{radegs} &
0.46 &
\cellcolor{yellow!25}0.66 &
\cellcolor{yellow!25}0.35 &
0.41 &
\cellcolor{orange!25}0.79 &
0.73 &
\cellcolor{yellow!25}0.67 &
1.14 &
1.18 &
\cellcolor{yellow!25}0.61 &
0.64 &
0.87 &
0.37 &
0.69 &
0.48 &
0.67 &
$10$m
\\ \hline
Ours (SH) & 
\cellcolor{red!25}0.38 &
\cellcolor{orange!25}0.63 &
\cellcolor{red!25}0.30 &
\cellcolor{yellow!25}0.35 &
\cellcolor{orange!25}0.79 &
\cellcolor{orange!25}0.63 &
\cellcolor{orange!25}0.65 &
1.10 &
1.21 &
0.62 &
\cellcolor{orange!25}0.48 &
1.18 &
\cellcolor{red!25}0.33 &
\cellcolor{orange!25}0.42 &
\cellcolor{red!25}0.38 &
\cellcolor{yellow!25}0.63 &
$10$m \\ 
Ours (Latent) & 
\cellcolor{yellow!25}0.40 &
\cellcolor{red!25}0.59 &
0.39 &
0.38 &
\cellcolor{red!25}0.72 &
\cellcolor{red!25}0.59 &
\cellcolor{orange!25}0.65 &
\cellcolor{yellow!25}1.08 &
\cellcolor{red!25}0.93 &
\cellcolor{orange!25}0.59 &
\cellcolor{yellow!25}0.50 &
\cellcolor{orange!25}0.67 &
\cellcolor{orange!25}0.34 &
\cellcolor{yellow!25}0.47 &
\cellcolor{yellow!25}0.40 &
\cellcolor{red!25}0.58 &
$11$m
 \\ 
\hline
\end{tabular}
\caption{Quantitative evalution on the DTU dataset based on the Chamfer Distance. 
Throughout the paper, the best metric is highlighted in red, the second best metric is highlighted in orange, and the third best metric is highlighted in yellow.
}
\label{table:dtu}
\end{table*}

We notice that, as shown in \cref{fig:color_merge}, if the color also continuously changes when the Gaussian surfel moves in such a way that whenever two Gaussian surfels intersect the ray at the exact same position, they have the same color, this discontinuity is solved. Therefore, given a ray which intersects with Gaussian surfels $\mathcal{G}_{1}, \mathcal{G}_{2}, ..., \mathcal{G}_{N}$, we propose to replace the color $c_i$ of $\mathcal{G}_i$ with $\widehat{c}_i$ as:
\begin{equation}
\label{equation-color-ours}
    \widehat{c}_i = \frac{\sum_{j=1}^N (1-\exp(-w_j)) \exp(-\tau|t_j - t_i|) c_j}{\sum_{j=1}^N (1-\exp(-w_j)) \exp(-\tau|t_j - t_i|)}, 
\end{equation}
where $\tau = 100$, and we blend the color based on the weight and distance.

\begin{theorem}
    With the color defined in \cref{equation-color-ours}, \cref{equation-color-ours-final} gives a continuous function of properties of Gaussian surfels.
\end{theorem}
\begin{proof}
    $\widehat{c}_i$ is a continuous function of $t_1, t_2, ..., t_N$. And when there exists an $i'\in\mathbb{N}$, such that $t_i=t_{i'}$, $\widehat{c}_i=\widehat{c}_{i'}$.
\end{proof}

With the color defined in \cref{equation-color-ours}, we also reach an almost exact rendering algorithm as discussed in \cref{section-method-geo-field-splatting}. 
However, it is computationally inefficient to blend the color per ray. We instead propose an efficient approximation that blends the color in the $\mathbb{R}^3$ space. Specifically, we have:
\begin{equation}
\label{equation-color-ours-approx}
    \widehat{c}_i = \frac{\sum_{j=1}^N (1-\exp(-w_j)) \exp(-\tau||\mathbf{m}_i - \mathbf{m}_j||_2) c_j}{\sum_{j=1}^N (1-\exp(-w_j)) \exp(-\tau||\mathbf{m}_i - \mathbf{m}_j||_2)}, 
\end{equation}
which can be efficiently implemented as propagating the color before rendering.

\subsection{Improve Color Representation}
\label{section-method-color-representation}
The method described above %
significantly improves surface reconstruction quality using spherical harmonics as the color representation. However, we find it struggles to reconstruct the specular surfaces, as observed in \cite{liu2023nero, refnerf, gaussianshader, dipole, 3dgsdr}. %

We instead assign a latent $\bm l_i \in\mathbb{R}^{32}$ to the $i^\text{th}$ Gaussian surfel and create a shallow MLP $\Phi$, such that:
\begin{equation}
    \mathbf{c}_i = \Phi(\bm l_i, \text{SE}(\bm \omega), \text{SE}(\bm \omega_o)), 
\end{equation}
where $\text{SE}(\cdot)$ denotes the spherical harmonics encoding \cite{sh}, $\bm \omega$ denotes the ray direction, and $\bm \omega_o$ denotes the reflected ray direction with respect to the local normal vector, i.e., $\bm \omega = (\mathbf{m}_i - \mathbf{o}) / || \mathbf{m}_i - \mathbf{o} ||_2$, and $\bm \omega_o = 2 (\mathbf{n}_i \cdot \bm \omega)\mathbf{n}_i - \bm \omega$. 
We find it beneficial to input both $\bm \omega$ and $\bm \omega_o$ for dealing with both diffuse and specular surfaces{, as in \cite{wang2025unisdf, verbin2024nerf}}.
\begin{table}[t]
\centering
\footnotesize
\setlength{\tabcolsep}{3pt}
\begin{tabular}{r|cccccccc|c}
\hline
\multicolumn{1}{c|}{\textbf{Methods}} & \textbf{EvaUni} & \textbf{Temple} & \textbf{Excava} & \textbf{Museum} & \textbf{Avg. $\downarrow$} \\
\hline
F.D.S. \cite{dipole} &
N/A & 
\cellcolor{red!25}1.87 & 
1.08 & 
\cellcolor{red!25}1.92 & 
\cellcolor{orange!25}1.62
\\
N.A. \cite{neuralangelo} & 
\cellcolor{orange!25}1.49 &
2.20 & 
0.88 & 
\cellcolor{yellow!25}2.07 & 
1.66
\\
\hline
2DGS \cite{2dgs} & 
1.56 & 
2.10 & 
0.98 & 
2.21 & 
1.71
\\
G.S. \cite{gaussian_surfel} & 
2.18 & 
2.26 & 
N/A & 
2.88 & 
2.44
\\
GOF \cite{GOF} & 
\cellcolor{yellow!25}1.55 & 
\cellcolor{yellow!25}1.99 & 
\cellcolor{orange!25}0.81 & 
2.23 & 
\cellcolor{yellow!25}1.65
\\
RaDe-GS \cite{radegs} & 
1.67 & 
2.01 & 
0.88 & 
2.19 & 
1.69
\\ 
\hline
Ours (SH) & 
\cellcolor{red!25}1.48 & 
\cellcolor{orange!25}1.89 & 
\cellcolor{red!25}0.72 & 
\cellcolor{orange!25}1.98 & 
\cellcolor{red!25}1.52
\\
Ours (Latent) & 
1.58 & 
2.13 & 
\cellcolor{yellow!25}0.87 & 
2.26 & 
1.71
\\
\hline
\end{tabular}
\caption{Quantitative comparison on the BlendedMVS dataset for scene-level cases based on the Chamfer Distance ($\times10^{-1}$). ``N/A'' denotes that the method fails to converge. }
\label{table:bmvs-scene}
\end{table}

\section{Results}
\label{sec:experiment}
We discuss implementation details and evaluation here. For detailed proof, further implementation details, evaluation results, and discussions, please refer to the supplementary.

\paragraph{Implementation Details.} We base our implementations and hyper-parameters on \cite{2dgs} and use \cite{tiny-cuda-nn} for implementing the MLP. We replace the original densification strategy \cite{2dgs, 3dgs} with \cite{ye2024absgs}. We use the TSDF fusion \cite{open3d} to extract the mesh from rendered depth maps. The proposed color blending over all kernels is approximated through closest-k-points \cite{ravi2020pytorch3d}, where $k=10$. All the experiments are conducted on an NVIDIA 6000 Ada.

\begin{table}[t]
\centering
\footnotesize
\setlength{\tabcolsep}{3pt}
\begin{tabular}{r|cccccccc|c}
\hline
\multicolumn{1}{c|}{\textbf{Methods}} & \textbf{Bea} & \textbf{Clo} & \textbf{Dog} & \textbf{Dur} & \textbf{Jad} & \textbf{Man} & \textbf{Scu} & \textbf{Sto} & \textbf{Avg. $\downarrow$} \\
\hline
F.D.S. \cite{dipole} & 
\cellcolor{orange!25}0.33 & 
\cellcolor{red!25}0.10 & 
\cellcolor{red!25}0.21 & 
\cellcolor{orange!25}2.63 & 
0.19 & 
0.80 & 
\cellcolor{yellow!25}0.44 & 
0.91 & 
\cellcolor{yellow!25}0.70
\\
N.A. \cite{neuralangelo} & 
\cellcolor{red!25}0.32 & 
\cellcolor{orange!25}0.12 & 
0.31 & 
3.11 & 
\cellcolor{yellow!25}0.15 & 
0.80 & 
\cellcolor{orange!25}0.43 & 
\cellcolor{orange!25}0.77 &
0.75
\\
\hline
2DGS \cite{2dgs} & 
1.17 & 
0.46 & 
0.45 & 
2.87 & 
\cellcolor{yellow!25}0.15 & 
0.53 & 
0.59 & 
0.90 & 
0.89
\\
G.S. \cite{gaussian_surfel} & 
0.57 & 
0.35 & 
0.27 & 
6.62 & 
0.16 & 
0.53 & 
0.76 & 
1.45 & 
1.34
\\
GOF \cite{GOF} & 
0.59 & 
0.26 & 
0.34 & 
\cellcolor{yellow!25}2.72 & 
\cellcolor{red!25}0.12 & 
0.53 & 
0.48 & 
1.03 & 
0.76
\\
RaDe-GS \cite{radegs} & 
\cellcolor{yellow!25}0.51 & 
0.25 & 
\cellcolor{orange!25}0.24 & 
2.83 & 
\cellcolor{red!25}0.12 & 
\cellcolor{yellow!25}0.52 & 
\cellcolor{orange!25}0.43 & 
0.84 & 
0.72
\\ 
\hline
Ours (SH) & 
0.55 & 
0.17 & 
0.40 & 
\cellcolor{orange!25}2.63 & 
\cellcolor{orange!25}0.14 & 
\cellcolor{red!25}0.43 & 
\cellcolor{red!25}0.36 & 
\cellcolor{red!25}0.71 & 
\cellcolor{orange!25}0.67
\\
Ours (Latent) & 
0.60 & 
\cellcolor{yellow!25}0.13 & 
\cellcolor{yellow!25}0.25 & 
\cellcolor{red!25}2.58 & 
\cellcolor{orange!25}0.14 & 
\cellcolor{orange!25}0.44 & 
\cellcolor{red!25}0.36 & 
\cellcolor{yellow!25}0.80 & 
\cellcolor{red!25}0.66
\\
\hline
\end{tabular}
\caption{Quantitative comparison on the BlendedMVS dataset for object-centric cases based on the Chamfer Distance ($\times10^{-2}$). }
\label{table:bmvs}
\end{table}

\begin{figure}[t]
    \centering
    \includegraphics[width=\linewidth]{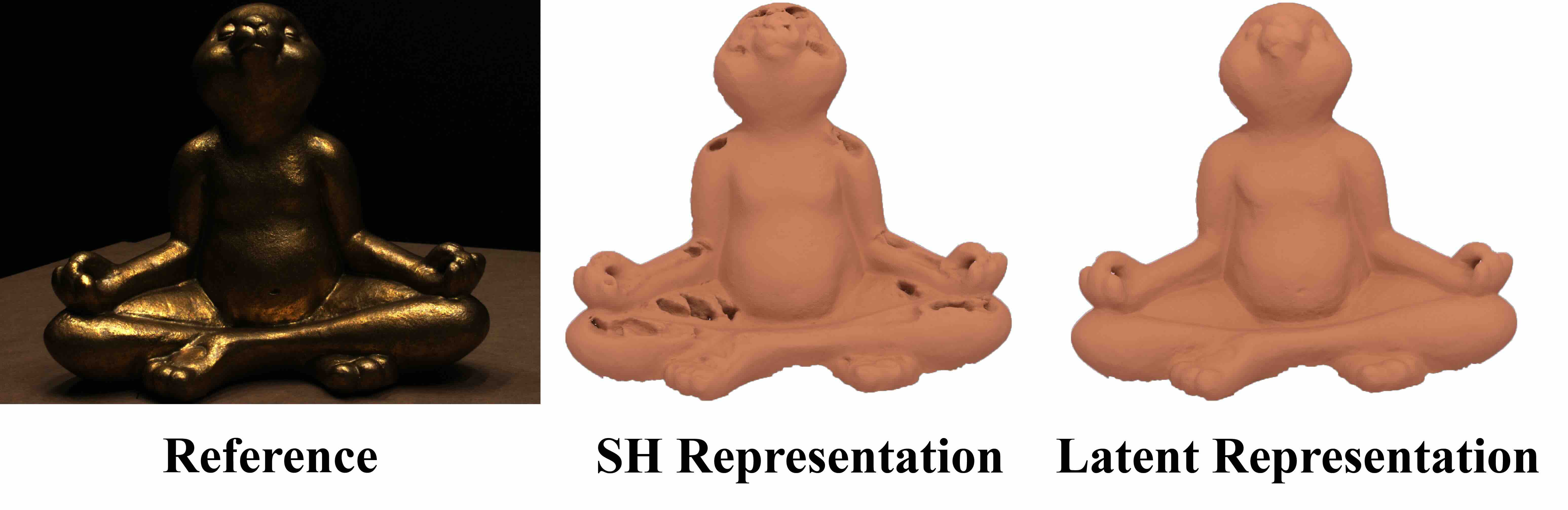}
    \caption{{We demonstrate the differences between SH and latent color representations on the {``scan110''} sample of the DTU dataset. When reconstructing specular surfaces (left), the SH representation (middle) leads to visible holes in the recovered mesh; these artifacts are resolved when using a latent representation (right).}}
    \label{fig:ablation_color_representation}
\end{figure}

\subsection{Qualitative Evaluation}
As shown in \cref{fig:qualitative}, we compare with \citet{dipole}, Neuralangelo \cite{neuralangelo}, RaDe-GS \cite{radegs}, and 2DGS \cite{2dgs} as representative methods on DTU \cite{dtu}, and BlendedMVS \cite{yao2020blendedmvs}. Our method better captures geometric details (e.g., nose of the skull, door of the house), while handling specular surfaces well (e.g., reflective cans and scissors), compared to \citet{dipole} and RaDe-GS. Our method also ensures smooth geometry without holes (e.g., dog, sculpture and temple). Neuralangelo fails to properly represent open-boundary objects (e.g., card held by the bear).
{We further prepare a preliminary experiment on highly specular surfaces in C.2 of the supplementary.}

\subsection{Quantitative Evaluation}
We %
evaluate on the datasets with ground-truth geometries and their corresponding camera parameters. 
We compare with methods, including NeuS2 \cite{neus2}, Neuralangelo \cite{neuralangelo}, \citet{object_as_volumes}, NeUDF \cite{neudf} as representative neural methods, SuGaR \cite{guedon2023sugar}, 2DGS \cite{2dgs}, \citet{gaussian_surfel}, GOF \cite{GOF}, RaDe-GS \cite{radegs} as representative splatting methods, and \citet{dipole} on the DTU \cite{dtu} dataset in \cref{table:dtu}. 
Besides, we compare with \citet{dipole}, Neuralangelo, 2DGS, \citet{gaussian_surfel}, GOF and RaDe-GS on the BlendedMVS \cite{yao2020blendedmvs} dataset for scene-level cases in \cref{table:bmvs-scene} and object-centric cases in \cref{table:bmvs}. 
We unify the evaluation protocol and do not apply masking loss during training, nor %
iterative closest point during evaluation. %

It is noteworthy that our method has the best averaged metrics. We achieve the best or second best results across almost all diffuse scenes while using the SH representation. For specular cases (i.e., ``scan97'', ``scan110'' of DTU, and ``Dog'', ``Clock'' of BlendedMVS), our method with improved color representation significantly improves the reconstruction quality and achieves one of the top three metrics. 
We are also significantly faster than other competing methods, including, Neuralangelo and \citet{dipole}.
The additional computational cost compared to 2DGS comes from our color blending which relies on the closest-k-points.

\subsection{Ablation Study}
We evaluate our proposed components on the DTU dataset \cite{dtu} with both SH representation and latent representation for colors, and report the Chamfer Distance (denoted as CD) in \cref{table:dtu-ablation}. 
Specifically, we replace the proposed geometry field-based splatting algorithm with the original 2DGS splatting algorithm, which leads to a severe drop of geometry quality. 
Besides, we also remove the proposed color propagation used for remedying loss landscape defects, which also leads to a drop in geometry quality. 
Furthermore, we study the effects of the only approximation we make, i.e., using global sorting to approximate the per-ray sorting. We implement a per-ray priority queue to sort the surfels based on the depths they intersect with the ray. Even though the per-ray sorting further improves the geometry quality, it consumes much more time ($3-4\times$) and memory. We then do not use per-ray sorting by default and in other experiments.

Moreover, we also demonstrate the benefits of using the latent representation for colors on the specular surfaces in \cref{fig:ablation_color_representation}. With SH representation, there are holes on the specular surface, while latent representation is free of this problem.

\begin{table}[t]
\centering
\footnotesize
\begin{tabular}{l|c|c}
\hline
\textbf{Methods}           & \textbf{Ours (SH)} & \textbf{Ours (Latent)} \\ \hline
Full Model & \cellcolor{orange!25}{0.63} & \cellcolor{orange!25}{0.58} \\
\quad- Geometry Field Splatting & 0.70 & 0.72       \\
\quad- Remedy Loss Landscape Defects   & \cellcolor{yellow!25}{0.67} & \cellcolor{yellow!25}{0.59}       \\ 
\quad+ Per-Ray Sorting & \cellcolor{red!25}{0.61} & \cellcolor{red!25}{0.57} \\
\hline
\end{tabular}
\caption{Ablation study for the proposed components on the DTU dataset based on the Chamfer {Distance}. %
}
\label{table:dtu-ablation}
\end{table}

\section{Conclusion}
\label{sec:conclusion}
In this paper, we propose to first define a stochastic geometry field in the space with Gaussian surfels, then convert it into the density field, and render it with an efficient and differentiable refined splatting algorithm for accurate geometry reconstruction. We also identify the discontinuity of view synthesis loss and propose an efficient remedy to address it such that the represented stochastic geometry can converge to the deterministic geometry well through the optimization.
Our method is not without limitations. %
For example, our method cannot handle transparent or semi-transparent objects due to the opaque assumption, or fuzzy objects well due to the smoothness constraint. 

In conclusion, we show that it is possible to use Gaussian surfels to accurately define the geometry and then perform efficient and differentiable almost exact rendering for geometry reconstruction. We achieve significant improvement over the geometry quality compared to other baselines, including both neural and splatting-based methods. %

\paragraph{Acknowledgements} This work is supported by the Intelligence Advanced Research Projects Activity (IARPA) via Department of Interior/ Interior Business Center (DOI/IBC) contract number 140D0423C0076. The U.S. Government is authorized to reproduce and distribute reprints for Governmental purposes notwithstanding any copyright annotation thereon. Disclaimer: The views and conclusions contained herein are those of the authors and should not be interpreted as necessarily representing the official policies or endorsements, either expressed or implied, of IARPA, DOI/IBC, or the U.S. Government. 
We also acknowledge ONR grant N00014-23-1-2526, and NSF grants 2100237 and 2120019 for the Nautilus cluster, gifts from Adobe, Google, Qualcomm and Rembrand, the Ronald L. Graham Chair and the UC San Diego Center for Visual Computing. We thank Yash Belhe for discussions, proofreading and feedback. We also thank the NVIDIA for GPU gifts.

{\small
\bibliographystyle{ieeenat_fullname}
\bibliography{11_references}

\begin{thebibliography}{74}
\providecommand{\natexlab}[1]{#1}
\providecommand{\url}[1]{\texttt{#1}}
\expandafter\ifx\csname urlstyle\endcsname\relax
  \providecommand{\doi}[1]{doi: #1}\else
  \providecommand{\doi}{doi: \begingroup \urlstyle{rm}\Url}\fi

\bibitem[Barron et~al.(2022)Barron, Mildenhall, Verbin, Srinivasan, and
  Hedman]{mipnerf360}
Jonathan~T Barron, Ben Mildenhall, Dor Verbin, Pratul~P Srinivasan, and Peter
  Hedman.
\newblock Mip-nerf 360: Unbounded anti-aliased neural radiance fields.
\newblock In \emph{Proceedings of the IEEE/CVF conference on computer vision
  and pattern recognition}, pages 5470--5479, 2022.

\bibitem[Bleyer et~al.(2011)Bleyer, Rhemann, and Rother]{Bleyer2011PatchMatchS}
Michael Bleyer, Christoph Rhemann, and Carsten Rother.
\newblock Patchmatch stereo - stereo matching with slanted support windows.
\newblock In \emph{British Machine Vision Conference}, 2011.

\bibitem[Botsch et~al.(2005)Botsch, Hornung, Zwicker, and
  Kobbelt]{10.5555/2386366.2386369}
Mario Botsch, Alexander Hornung, Matthias Zwicker, and Leif Kobbelt.
\newblock High-quality surface splatting on today's gpus.
\newblock In \emph{Proceedings of the Second Eurographics / IEEE VGTC
  Conference on Point-Based Graphics}, page 17–24, Goslar, DEU, 2005.
  Eurographics Association.

\bibitem[Cai et~al.(2023)Cai, Ji, Yan, and Xu]{riav}
Changjiang Cai, Pan Ji, Qingan Yan, and Yi Xu.
\newblock Riav-mvs: Recurrent-indexing an asymmetric volume for multi-view
  stereo.
\newblock In \emph{Proceedings of the IEEE/CVF conference on computer vision
  and pattern recognition}, pages 919--928, 2023.

\bibitem[Chen et~al.(2024{\natexlab{a}})Chen, Li, Ye, Wang, Xie, Zhai, Wang,
  Liu, Bao, and Zhang]{pgsr}
Danpeng Chen, Hai Li, Weicai Ye, Yifan Wang, Weijian Xie, Shangjin Zhai, Nan
  Wang, Haomin Liu, Hujun Bao, and Guofeng Zhang.
\newblock Pgsr: Planar-based gaussian splatting for efficient and high-fidelity
  surface reconstruction.
\newblock \emph{arXiv preprint arXiv:2406.06521}, 2024{\natexlab{a}}.

\bibitem[Chen et~al.(2024{\natexlab{b}})Chen, Oberson, Feldmann, Schreer,
  Hilsmann, and Eisert]{surfel-human-1}
Decai Chen, Brianne Oberson, Ingo Feldmann, Oliver Schreer, Anna Hilsmann, and
  Peter Eisert.
\newblock Adaptive and temporally consistent gaussian surfels for multi-view
  dynamic reconstruction.
\newblock \emph{arXiv preprint arXiv:2411.06602}, 2024{\natexlab{b}}.

\bibitem[Chen et~al.(2024{\natexlab{c}})Chen, Miller, and Gkioulekas]{dipole}
Hanyu Chen, Bailey Miller, and Ioannis Gkioulekas.
\newblock 3d reconstruction with fast dipole sums.
\newblock \emph{arXiv preprint arXiv:2405.16788}, 2024{\natexlab{c}}.

\bibitem[Curless and Levoy(2023)]{10.1145/3596711.3596726}
Brian Curless and Marc Levoy.
\newblock \emph{A Volumetric Method for Building Complex Models from Range
  Images}.
\newblock Association for Computing Machinery, New York, NY, USA, 1 edition,
  2023.

\bibitem[Dai et~al.(2024)Dai, Xu, Xie, Liu, Wang, and Xu]{gaussian_surfel}
Pinxuan Dai, Jiamin Xu, Wenxiang Xie, Xinguo Liu, Huamin Wang, and Weiwei Xu.
\newblock High-quality surface reconstruction using gaussian surfels.
\newblock In \emph{ACM SIGGRAPH 2024 Conference Papers}, pages 1--11, 2024.

\bibitem[Darmon et~al.(2022)Darmon, Bascle, Devaux, Monasse, and
  Aubry]{neuralwarp}
Fran{\c{c}}ois Darmon, B{\'e}n{\'e}dicte Bascle, Jean-Cl{\'e}ment Devaux,
  Pascal Monasse, and Mathieu Aubry.
\newblock Improving neural implicit surfaces geometry with patch warping.
\newblock In \emph{Proceedings of the IEEE/CVF Conference on Computer Vision
  and Pattern Recognition}, pages 6260--6269, 2022.

\bibitem[Ding et~al.(2022{\natexlab{a}})Ding, Yuan, Zhu, Zhang, Liu, Wang, and
  Liu]{ding2022transmvsnet}
Yikang Ding, Wentao Yuan, Qingtian Zhu, Haotian Zhang, Xiangyue Liu, Yuanjiang
  Wang, and Xiao Liu.
\newblock Transmvsnet: Global context-aware multi-view stereo network with
  transformers.
\newblock In \emph{Proceedings of the IEEE/CVF Conference on Computer Vision
  and Pattern Recognition}, pages 8585--8594, 2022{\natexlab{a}}.

\bibitem[Ding et~al.(2022{\natexlab{b}})Ding, Zhu, Liu, Yuan, Zhang, and
  Zhang]{kd-mvs}
Yikang Ding, Qingtian Zhu, Xiangyue Liu, Wentao Yuan, Haotian Zhang, and Chi
  Zhang.
\newblock Kd-mvs: Knowledge distillation based self-supervised learning for
  multi-view stereo.
\newblock In \emph{European conference on computer vision}, pages 630--646.
  Springer, 2022{\natexlab{b}}.

\bibitem[Dong et~al.(2024)Dong, Xu, Gao, Bao, Xu, and Lau]{surfel-human-2}
Zheng Dong, Ke Xu, Yaoan Gao, Hujun Bao, Weiwei Xu, and Rynson W.~H. Lau.
\newblock Gaussian surfel splatting for live human performance capture.
\newblock \emph{ACM Trans. Graph.}, 43\penalty0 (6), 2024.

\bibitem[Fu et~al.(2022)Fu, Xu, Ong, and Tao]{geoneus}
Qiancheng Fu, Qingshan Xu, Yew~Soon Ong, and Wenbing Tao.
\newblock Geo-neus: Geometry-consistent neural implicit surfaces learning for
  multi-view reconstruction.
\newblock \emph{Advances in Neural Information Processing Systems},
  35:\penalty0 3403--3416, 2022.

\bibitem[Furukawa and Ponce(2010)]{5226635}
Yasutaka Furukawa and Jean Ponce.
\newblock Accurate, dense, and robust multiview stereopsis.
\newblock \emph{IEEE Transactions on Pattern Analysis and Machine
  Intelligence}, 32\penalty0 (8):\penalty0 1362--1376, 2010.

\bibitem[Ge et~al.(2023)Ge, Hu, Zhao, Liu, and Chen]{ge2023ref}
Wenhang Ge, Tao Hu, Haoyu Zhao, Shu Liu, and Ying-Cong Chen.
\newblock Ref-neus: Ambiguity-reduced neural implicit surface learning for
  multi-view reconstruction with reflection.
\newblock In \emph{Proceedings of the IEEE/CVF International Conference on
  Computer Vision}, pages 4251--4260, 2023.

\bibitem[Goesele et~al.(2007)Goesele, Snavely, Curless, Hoppe, and
  Seitz]{multi-view-stereo}
Michael Goesele, Noah Snavely, Brian Curless, Hugues Hoppe, and Steven~M.
  Seitz.
\newblock Multi-view stereo for community photo collections.
\newblock In \emph{2007 IEEE 11th International Conference on Computer Vision},
  pages 1--8, 2007.

\bibitem[Gropp et~al.(2020)Gropp, Yariv, Haim, Atzmon, and Lipman]{eik_loss}
Amos Gropp, Lior Yariv, Niv Haim, Matan Atzmon, and Yaron Lipman.
\newblock Implicit geometric regularization for learning shapes.
\newblock \emph{arXiv preprint arXiv:2002.10099}, 2020.

\bibitem[Gu et~al.(2020)Gu, Fan, Zhu, Dai, Tan, and Tan]{gu2020cascade}
Xiaodong Gu, Zhiwen Fan, Siyu Zhu, Zuozhuo Dai, Feitong Tan, and Ping Tan.
\newblock Cascade cost volume for high-resolution multi-view stereo and stereo
  matching.
\newblock In \emph{Proceedings of the IEEE/CVF Conference on Computer Vision
  and Pattern Recognition}, pages 2495--2504, 2020.

\bibitem[Gu{\'e}don and Lepetit(2024)]{guedon2023sugar}
Antoine Gu{\'e}don and Vincent Lepetit.
\newblock Sugar: Surface-aligned gaussian splatting for efficient 3d mesh
  reconstruction and high-quality mesh rendering.
\newblock \emph{CVPR}, 2024.

\bibitem[Hirschmuller(2008)]{4359315}
Heiko Hirschmuller.
\newblock Stereo processing by semiglobal matching and mutual information.
\newblock \emph{IEEE Transactions on Pattern Analysis and Machine
  Intelligence}, 30\penalty0 (2):\penalty0 328--341, 2008.

\bibitem[Huang et~al.(2024)Huang, Yu, Chen, Geiger, and Gao]{2dgs}
Binbin Huang, Zehao Yu, Anpei Chen, Andreas Geiger, and Shenghua Gao.
\newblock 2d gaussian splatting for geometrically accurate radiance fields.
\newblock In \emph{ACM SIGGRAPH 2024 Conference Papers}, pages 1--11, 2024.

\bibitem[Jensen et~al.(2014)Jensen, Dahl, Vogiatzis, Tola, and Aanæs]{dtu}
Rasmus Jensen, Anders Dahl, George Vogiatzis, Engil Tola, and Henrik Aanæs.
\newblock Large scale multi-view stereopsis evaluation.
\newblock In \emph{2014 IEEE Conference on Computer Vision and Pattern
  Recognition}, pages 406--413, 2014.

\bibitem[Jiang et~al.(2024)Jiang, Tu, Liu, Gao, Long, Wang, and
  Ma]{gaussianshader}
Yingwenqi Jiang, Jiadong Tu, Yuan Liu, Xifeng Gao, Xiaoxiao Long, Wenping Wang,
  and Yuexin Ma.
\newblock Gaussianshader: 3d gaussian splatting with shading functions for
  reflective surfaces.
\newblock In \emph{Proceedings of the IEEE/CVF Conference on Computer Vision
  and Pattern Recognition}, pages 5322--5332, 2024.

\bibitem[Kerbl et~al.(2023)Kerbl, Kopanas, Leimk{\"u}hler, and Drettakis]{3dgs}
Bernhard Kerbl, Georgios Kopanas, Thomas Leimk{\"u}hler, and George Drettakis.
\newblock 3d gaussian splatting for real-time radiance field rendering.
\newblock \emph{ACM Trans. Graph.}, 42\penalty0 (4):\penalty0 139--1, 2023.

\bibitem[Knapitsch et~al.(2017)Knapitsch, Park, Zhou, and Koltun]{tanks}
Arno Knapitsch, Jaesik Park, Qian-Yi Zhou, and Vladlen Koltun.
\newblock Tanks and temples: Benchmarking large-scale scene reconstruction.
\newblock \emph{ACM Transactions on Graphics}, 36\penalty0 (4), 2017.

\bibitem[Laga et~al.(2020)Laga, Jospin, Boussaid, and
  Bennamoun]{laga2020survey}
Hamid Laga, Laurent~Valentin Jospin, Farid Boussaid, and Mohammed Bennamoun.
\newblock A survey on deep learning techniques for stereo-based depth
  estimation.
\newblock \emph{IEEE transactions on pattern analysis and machine
  intelligence}, 44\penalty0 (4):\penalty0 1738--1764, 2020.

\bibitem[Li et~al.(2023)Li, M\"uller, Evans, Taylor, Unberath, Liu, and
  Lin]{neuralangelo}
Zhaoshuo Li, Thomas M\"uller, Alex Evans, Russell~H Taylor, Mathias Unberath,
  Ming-Yu Liu, and Chen-Hsuan Lin.
\newblock Neuralangelo: High-fidelity neural surface reconstruction.
\newblock In \emph{IEEE Conference on Computer Vision and Pattern Recognition
  ({CVPR})}, 2023.

\bibitem[Liu et~al.(2023{\natexlab{a}})Liu, Wang, Lin, Long, Wang, Liu, Komura,
  and Wang]{liu2023nero}
Yuan Liu, Peng Wang, Cheng Lin, Xiaoxiao Long, Jiepeng Wang, Lingjie Liu, Taku
  Komura, and Wenping Wang.
\newblock Nero: Neural geometry and brdf reconstruction of reflective objects
  from multiview images.
\newblock \emph{ACM Transactions on Graphics (TOG)}, 42\penalty0 (4):\penalty0
  1--22, 2023{\natexlab{a}}.

\bibitem[Liu et~al.(2023{\natexlab{b}})Liu, Wang, Yang, Chen, Meng, Yang, and
  Gao]{neudf}
Yu-Tao Liu, Li Wang, Jie Yang, Weikai Chen, Xiaoxu Meng, Bo Yang, and Lin Gao.
\newblock Neudf: Leaning neural unsigned distance fields with volume rendering.
\newblock In \emph{Proceedings of the IEEE/CVF Conference on Computer Vision
  and Pattern Recognition}, pages 237--247, 2023{\natexlab{b}}.

\bibitem[Long et~al.(2022)Long, Lin, Wang, Komura, and
  Wang]{long2022sparseneus}
Xiaoxiao Long, Cheng Lin, Peng Wang, Taku Komura, and Wenping Wang.
\newblock Sparseneus: Fast generalizable neural surface reconstruction from
  sparse views.
\newblock In \emph{European Conference on Computer Vision}, pages 210--227.
  Springer, 2022.

\bibitem[Long et~al.(2023)Long, Lin, Liu, Liu, Wang, Theobalt, Komura, and
  Wang]{neuraludf}
Xiaoxiao Long, Cheng Lin, Lingjie Liu, Yuan Liu, Peng Wang, Christian Theobalt,
  Taku Komura, and Wenping Wang.
\newblock Neuraludf: Learning unsigned distance fields for multi-view
  reconstruction of surfaces with arbitrary topologies.
\newblock In \emph{Proceedings of the IEEE/CVF Conference on Computer Vision
  and Pattern Recognition}, pages 20834--20843, 2023.

\bibitem[Mildenhall et~al.(2020)Mildenhall, Srinivasan, Tancik, Barron,
  Ramamoorthi, and Ng]{nerf2020}
Ben Mildenhall, Pratul~P. Srinivasan, Matthew Tancik, Jonathan~T. Barron, Ravi
  Ramamoorthi, and Ren Ng.
\newblock Nerf: Representing scenes as neural radiance fields for view
  synthesis.
\newblock In \emph{ECCV}, 2020.

\bibitem[Miller et~al.(2024)Miller, Chen, Lai, and
  Gkioulekas]{object_as_volumes}
Bailey Miller, Hanyu Chen, Alice Lai, and Ioannis Gkioulekas.
\newblock Objects as volumes: A stochastic geometry view of opaque solids.
\newblock In \emph{Proceedings of the IEEE/CVF Conference on Computer Vision
  and Pattern Recognition (CVPR)}, pages 87--97, 2024.

\bibitem[M\"uller(2021)]{tiny-cuda-nn}
Thomas M\"uller.
\newblock {tiny-cuda-nn}, 2021.

\bibitem[M{\"u}ller et~al.(2022)M{\"u}ller, Evans, Schied, and
  Keller]{instant-ngp}
Thomas M{\"u}ller, Alex Evans, Christoph Schied, and Alexander Keller.
\newblock Instant neural graphics primitives with a multiresolution hash
  encoding.
\newblock \emph{ACM transactions on graphics (TOG)}, 41\penalty0 (4):\penalty0
  1--15, 2022.

\bibitem[Pfister et~al.(2000)Pfister, Zwicker, van Baar, and Gross]{Surfels}
Hanspeter Pfister, Matthias Zwicker, Jeroen van Baar, and Markus Gross.
\newblock Surfels: surface elements as rendering primitives.
\newblock In \emph{Proceedings of the 27th Annual Conference on Computer
  Graphics and Interactive Techniques}, page 335–342, USA, 2000. ACM
  Press/Addison-Wesley Publishing Co.

\bibitem[Radl et~al.(2024)Radl, Steiner, Parger, Weinrauch, Kerbl, and
  Steinberger]{radl2024stopthepop}
Lukas Radl, Michael Steiner, Mathias Parger, Alexander Weinrauch, Bernhard
  Kerbl, and Markus Steinberger.
\newblock Stopthepop: Sorted gaussian splatting for view-consistent real-time
  rendering.
\newblock \emph{ACM Transactions on Graphics (TOG)}, 43\penalty0 (4):\penalty0
  1--17, 2024.

\bibitem[Ramamoorthi and Hanrahan(2001)]{sh}
Ravi Ramamoorthi and Pat Hanrahan.
\newblock An efficient representation for irradiance environment maps.
\newblock In \emph{Proceedings of the 28th Annual Conference on Computer
  Graphics and Interactive Techniques}, page 497–500, New York, NY, USA,
  2001. Association for Computing Machinery.

\bibitem[Ravi et~al.(2020)Ravi, Reizenstein, Novotny, Gordon, Lo, Johnson, and
  Gkioxari]{ravi2020pytorch3d}
Nikhila Ravi, Jeremy Reizenstein, David Novotny, Taylor Gordon, Wan-Yen Lo,
  Justin Johnson, and Georgia Gkioxari.
\newblock Accelerating 3d deep learning with pytorch3d.
\newblock \emph{arXiv:2007.08501}, 2020.

\bibitem[Reiser et~al.(2023)Reiser, Szeliski, Verbin, Srinivasan, Mildenhall,
  Geiger, Barron, and Hedman]{reiser2023merf}
Christian Reiser, Rick Szeliski, Dor Verbin, Pratul Srinivasan, Ben Mildenhall,
  Andreas Geiger, Jon Barron, and Peter Hedman.
\newblock Merf: Memory-efficient radiance fields for real-time view synthesis
  in unbounded scenes.
\newblock \emph{ACM Transactions on Graphics (TOG)}, 42\penalty0 (4):\penalty0
  1--12, 2023.

\bibitem[Reiser et~al.(2024)Reiser, Garbin, Srinivasan, Verbin, Szeliski,
  Mildenhall, Barron, Hedman, and Geiger]{BOG}
Christian Reiser, Stephan Garbin, Pratul~P. Srinivasan, Dor Verbin, Richard
  Szeliski, Ben Mildenhall, Jonathan~T. Barron, Peter Hedman, and Andreas
  Geiger.
\newblock Binary opacity grids: Capturing fine geometric detail for mesh-based
  view synthesis.
\newblock \emph{SIGGRAPH}, 2024.

\bibitem[Ren et~al.(2024)Ren, Cao, Fu, and Xue]{mvsneus}
Xinlin Ren, Chenjie Cao, Yanwei Fu, and Xiangyang Xue.
\newblock Improving neural surface reconstruction with feature priors from
  multi-view image.
\newblock \emph{arXiv preprint arXiv:2408.02079}, 2024.

\bibitem[Sch{\"o}nberger et~al.(2016)Sch{\"o}nberger, Zheng, Frahm, and
  Pollefeys]{multi-view-stereo-2}
Johannes~L. Sch{\"o}nberger, Enliang Zheng, Jan-Michael Frahm, and Marc
  Pollefeys.
\newblock Pixelwise view selection for unstructured multi-view stereo.
\newblock In \emph{Computer Vision -- ECCV 2016}, pages 501--518, Cham, 2016.
  Springer International Publishing.

\bibitem[Sun et~al.(2003)Sun, Zheng, and Shum]{1206509}
Jian Sun, Nan-Ning Zheng, and Heung-Yeung Shum.
\newblock Stereo matching using belief propagation.
\newblock \emph{IEEE Transactions on Pattern Analysis and Machine
  Intelligence}, 25\penalty0 (7):\penalty0 787--800, 2003.

\bibitem[Tang(2022)]{torch-ngp}
Jiaxiang Tang.
\newblock Torch-ngp: a pytorch implementation of instant-ngp, 2022.
\newblock https://github.com/ashawkey/torch-ngp.

\bibitem[Verbin et~al.(2022)Verbin, Hedman, Mildenhall, Zickler, Barron, and
  Srinivasan]{refnerf}
Dor Verbin, Peter Hedman, Ben Mildenhall, Todd Zickler, Jonathan~T Barron, and
  Pratul~P Srinivasan.
\newblock Ref-nerf: Structured view-dependent appearance for neural radiance
  fields.
\newblock In \emph{2022 IEEE/CVF Conference on Computer Vision and Pattern
  Recognition (CVPR)}, pages 5481--5490. IEEE, 2022.

\bibitem[Verbin et~al.(2024)Verbin, Srinivasan, Hedman, Mildenhall, Attal,
  Szeliski, and Barron]{verbin2024nerf}
Dor Verbin, Pratul~P Srinivasan, Peter Hedman, Ben Mildenhall, Benjamin Attal,
  Richard Szeliski, and Jonathan~T Barron.
\newblock Nerf-casting: Improved view-dependent appearance with consistent
  reflections.
\newblock In \emph{SIGGRAPH Asia 2024 Conference Papers}, pages 1--10, 2024.

\bibitem[Vu et~al.(2012)Vu, Labatut, Pons, and Keriven]{5989831}
Hoang-Hiep Vu, Patrick Labatut, Jean-Philippe Pons, and Renaud Keriven.
\newblock High accuracy and visibility-consistent dense multiview stereo.
\newblock \emph{IEEE Transactions on Pattern Analysis and Machine
  Intelligence}, 34\penalty0 (5):\penalty0 889--901, 2012.

\bibitem[Wang et~al.(2021{\natexlab{a}})Wang, Galliani, Vogel, Speciale, and
  Pollefeys]{patchmatchnet}
Fangjinhua Wang, Silvano Galliani, Christoph Vogel, Pablo Speciale, and Marc
  Pollefeys.
\newblock Patchmatchnet: Learned multi-view patchmatch stereo.
\newblock In \emph{Proceedings of the IEEE/CVF conference on computer vision
  and pattern recognition}, pages 14194--14203, 2021{\natexlab{a}}.

\bibitem[Wang et~al.(2025)Wang, Rakotosaona, Niemeyer, Szeliski, Pollefeys, and
  Tombari]{wang2025unisdf}
Fangjinhua Wang, Marie-Julie Rakotosaona, Michael Niemeyer, Richard Szeliski,
  Marc Pollefeys, and Federico Tombari.
\newblock Unisdf: Unifying neural representations for high-fidelity 3d
  reconstruction of complex scenes with reflections.
\newblock \emph{Advances in Neural Information Processing Systems},
  37:\penalty0 3157--3184, 2025.

\bibitem[Wang et~al.(2021{\natexlab{b}})Wang, Liu, Liu, Theobalt, Komura, and
  Wang]{neus}
Peng Wang, Lingjie Liu, Yuan Liu, Christian Theobalt, Taku Komura, and Wenping
  Wang.
\newblock Neus: Learning neural implicit surfaces by volume rendering for
  multi-view reconstruction.
\newblock \emph{NeurIPS}, 2021{\natexlab{b}}.

\bibitem[Wang et~al.(2022)Wang, Li, and Dai]{effimvs}
Shaoqian Wang, Bo Li, and Yuchao Dai.
\newblock Efficient multi-view stereo by iterative dynamic cost volume.
\newblock In \emph{Proceedings of the IEEE/CVF Conference on Computer Vision
  and Pattern Recognition}, pages 8655--8664, 2022.

\bibitem[Wang et~al.(2021{\natexlab{c}})Wang, Wang, Liu, Zhou, Zhang, Zheng,
  and Bai]{WANG2021102102}
Xiang Wang, Chen Wang, Bing Liu, Xiaoqing Zhou, Liang Zhang, Jin Zheng, and
  Xiao Bai.
\newblock Multi-view stereo in the deep learning era: A comprehensive review.
\newblock \emph{Displays}, 70:\penalty0 102102, 2021{\natexlab{c}}.

\bibitem[Wang et~al.(2023)Wang, Han, Habermann, Daniilidis, Theobalt, and
  Liu]{neus2}
Yiming Wang, Qin Han, Marc Habermann, Kostas Daniilidis, Christian Theobalt,
  and Lingjie Liu.
\newblock Neus2: Fast learning of neural implicit surfaces for multi-view
  reconstruction.
\newblock In \emph{Proceedings of the IEEE/CVF International Conference on
  Computer Vision}, pages 3295--3306, 2023.

\bibitem[Wolf et~al.(2025)Wolf, Bracha, and Kimmel]{wolf2024gsmesh}
Yaniv Wolf, Amit Bracha, and Ron Kimmel.
\newblock {GS}2{M}esh: Surface reconstruction from {G}aussian splatting via
  novel stereo views.
\newblock In \emph{European Conference on Computer Vision}, pages 207--224.
  Springer, 2025.

\bibitem[Wu et~al.(2022)Wu, Wang, Pan, Xu, Theobalt, Liu, and Lin]{voxurf}
Tong Wu, Jiaqi Wang, Xingang Pan, Xudong Xu, Christian Theobalt, Ziwei Liu, and
  Dahua Lin.
\newblock Voxurf: Voxel-based efficient and accurate neural surface
  reconstruction.
\newblock \emph{arXiv preprint arXiv:2208.12697}, 2022.

\bibitem[Xiong et~al.(2023)Xiong, Peng, Zhang, Feng, Jiao, Gao, and
  Wang]{cl-mvsnet}
Kaiqiang Xiong, Rui Peng, Zhe Zhang, Tianxing Feng, Jianbo Jiao, Feng Gao, and
  Ronggang Wang.
\newblock Cl-mvsnet: Unsupervised multi-view stereo with dual-level contrastive
  learning.
\newblock In \emph{Proceedings of the IEEE/CVF International Conference on
  Computer Vision}, pages 3769--3780, 2023.

\bibitem[Yao et~al.(2018)Yao, Luo, Li, Fang, and Quan]{yao2018mvsnet}
Yao Yao, Zixin Luo, Shiwei Li, Tian Fang, and Long Quan.
\newblock Mvsnet: Depth inference for unstructured multi-view stereo.
\newblock In \emph{Proceedings of the European conference on computer vision
  (ECCV)}, pages 767--783, 2018.

\bibitem[Yao et~al.(2020)Yao, Luo, Li, Zhang, Ren, Zhou, Fang, and
  Quan]{yao2020blendedmvs}
Yao Yao, Zixin Luo, Shiwei Li, Jingyang Zhang, Yufan Ren, Lei Zhou, Tian Fang,
  and Long Quan.
\newblock Blendedmvs: A large-scale dataset for generalized multi-view stereo
  networks.
\newblock \emph{Computer Vision and Pattern Recognition (CVPR)}, 2020.

\bibitem[Yariv et~al.(2020)Yariv, Kasten, Moran, Galun, Atzmon, Ronen, and
  Lipman]{idr}
Lior Yariv, Yoni Kasten, Dror Moran, Meirav Galun, Matan Atzmon, Basri Ronen,
  and Yaron Lipman.
\newblock Multiview neural surface reconstruction by disentangling geometry and
  appearance.
\newblock \emph{Advances in Neural Information Processing Systems},
  33:\penalty0 2492--2502, 2020.

\bibitem[Yariv et~al.(2021)Yariv, Gu, Kasten, and Lipman]{volsdf}
Lior Yariv, Jiatao Gu, Yoni Kasten, and Yaron Lipman.
\newblock Volume rendering of neural implicit surfaces.
\newblock \emph{Advances in Neural Information Processing Systems},
  34:\penalty0 4805--4815, 2021.

\bibitem[Yariv et~al.(2023)Yariv, Hedman, Reiser, Verbin, Srinivasan, Szeliski,
  Barron, and Mildenhall]{BakedSDF}
Lior Yariv, Peter Hedman, Christian Reiser, Dor Verbin, Pratul~P. Srinivasan,
  Richard Szeliski, Jonathan~T. Barron, and Ben Mildenhall.
\newblock Bakedsdf: Meshing neural sdfs for real-time view synthesis.
\newblock In \emph{ACM SIGGRAPH 2023 Conference Proceedings}, New York, NY,
  USA, 2023. Association for Computing Machinery.

\bibitem[Ye et~al.(2024{\natexlab{a}})Ye, Hou, and Zhou]{3dgsdr}
Keyang Ye, Qiming Hou, and Kun Zhou.
\newblock 3d gaussian splatting with deferred reflection.
\newblock In \emph{ACM SIGGRAPH 2024 Conference Papers}, pages 1--10,
  2024{\natexlab{a}}.

\bibitem[Ye et~al.(2024{\natexlab{b}})Ye, Li, Liu, Qiao, and Dou]{ye2024absgs}
Zongxin Ye, Wenyu Li, Sidun Liu, Peng Qiao, and Yong Dou.
\newblock Absgs: Recovering fine details in 3d gaussian splatting.
\newblock In \emph{ACM Multimedia 2024}, 2024{\natexlab{b}}.

\bibitem[Yu and Gao(2020)]{yu2020fast}
Zehao Yu and Shenghua Gao.
\newblock Fast-mvsnet: Sparse-to-dense multi-view stereo with learned
  propagation and gauss-newton refinement.
\newblock In \emph{Proceedings of the IEEE/CVF conference on computer vision
  and pattern recognition}, pages 1949--1958, 2020.

\bibitem[Yu et~al.(2024{\natexlab{a}})Yu, Chen, Huang, Sattler, and
  Geiger]{Yu2024MipSplatting}
Zehao Yu, Anpei Chen, Binbin Huang, Torsten Sattler, and Andreas Geiger.
\newblock Mip-splatting: Alias-free 3d gaussian splatting.
\newblock In \emph{Proceedings of the IEEE/CVF Conference on Computer Vision
  and Pattern Recognition (CVPR)}, pages 19447--19456, 2024{\natexlab{a}}.

\bibitem[Yu et~al.(2024{\natexlab{b}})Yu, Sattler, and Geiger]{GOF}
Zehao Yu, Torsten Sattler, and Andreas Geiger.
\newblock Gaussian opacity fields: Efficient high-quality compact surface
  reconstruction in unbounded scenes.
\newblock \emph{arXiv:2404.10772}, 2024{\natexlab{b}}.

\bibitem[Zeng et~al.(2024)Zeng, Deschaintre, Georgiev, Hold-Geoffroy, Hu, Luan,
  Yan, and Ha{\v{s}}an]{zeng2024rgb}
Zheng Zeng, Valentin Deschaintre, Iliyan Georgiev, Yannick Hold-Geoffroy, Yiwei
  Hu, Fujun Luan, Ling-Qi Yan, and Milo{\v{s}} Ha{\v{s}}an.
\newblock Rgb$\leftrightarrow$x: Image decomposition and synthesis using
  material-and lighting-aware diffusion models.
\newblock In \emph{ACM SIGGRAPH 2024 Conference Papers}, pages 1--11, 2024.

\bibitem[Zhang et~al.(2024)Zhang, Fang, Shrestha, Liang, Long, and Tan]{radegs}
Baowen Zhang, Chuan Fang, Rakesh Shrestha, Yixun Liang, Xiaoxiao Long, and Ping
  Tan.
\newblock Rade-gs: Rasterizing depth in gaussian splatting.
\newblock \emph{arXiv preprint arXiv:2406.01467}, 2024.

\bibitem[Zhao et~al.(2024)Zhao, Wu, Huang, Zhi, Zhao, Wang, and
  Gao]{zhao2024surfel}
Yiqun Zhao, Chenming Wu, Binbin Huang, Yihao Zhi, Chen Zhao, Jingdong Wang, and
  Shenghua Gao.
\newblock Surfel-based gaussian inverse rendering for fast and relightable
  dynamic human reconstruction from monocular video.
\newblock \emph{arXiv preprint arXiv:2407.15212}, 2024.

\bibitem[Zhou et~al.(2018)Zhou, Park, and Koltun]{open3d}
Qian-Yi Zhou, Jaesik Park, and Vladlen Koltun.
\newblock Open3d: A modern library for 3d data processing.
\newblock \emph{arXiv preprint arXiv:1801.09847}, 2018.

\bibitem[Zhu et~al.(2021)Zhu, Min, Wei, Chen, and Wang]{zhu2021deep}
Qingtian Zhu, Chen Min, Zizhuang Wei, Yisong Chen, and Guoping Wang.
\newblock Deep learning for multi-view stereo via plane sweep: A survey.
\newblock \emph{arXiv preprint arXiv:2106.15328}, 2021.

\bibitem[Zwicker et~al.(2001)Zwicker, Pfister, van Baar, and
  Gross]{volume_splatting}
M. Zwicker, H. Pfister, J. van Baar, and M. Gross.
\newblock Ewa volume splatting.
\newblock In \emph{Proceedings Visualization, 2001. VIS '01.}, pages 29--538,
  2001.

\end{thebibliography}
}

\ifarxiv \clearpage \appendix \renewcommand{\baselinestretch}{1.00}\normalsize

\maketitlesupplementary

\appendix

\section{Derivations}
\subsection{Revisit Volume Splatting}
\subsubsection{Original Derivation}
Given a density field $\sigma(\mathbf{x)}$ and a color field $\mathbf{c}(\mathbf{x})$, where $\mathbf{x}\in\mathbb{R}^3$, the volume splatting algorithm \cite{volume_splatting} proposes to decompose the density field into the weighted sum of a set of $n$ independent kernels $\{\mathcal{K}_1, \mathcal{K}_2, ..., \mathcal{K}_n\}$, each of which is a function mapping from $x$ to a scalar and associated with a weight $\omega_i \in \mathbb{R}_{+}, i=1,2,...,n$.

Formally, the density field is expressed as:
\begin{equation}
\label{equation-decomposition}
    \sigma(\mathbf{x}) = \sum_{i=1}^n \omega_i \mathcal{K}_i(\mathbf{x}).
\end{equation}
This decomposition helps simplify the volume rendering equation, which does not consider the scattering, for efficient rendering.

Specifically, given a ray shooting from the camera origin $\mathbf{o}$ with direction $\mathbf{d}$, the location of a point $\mathbf{p}$ on the ray can be expressed as $\mathbf{p}(l) = \mathbf{o}+\mathbf{d}l$, where $l\in\mathbb{R}_{+}$ denotes the depth. The rendered color $\mathbf{C}$ with the exponential falloff is given by the following equation:
\begin{equation}
\label{equation-volume-rendering-color}
    \mathbf{C} = \int_{0}^{\infty} \mathbf{c}(\mathbf{p}(l)) \sigma(\mathbf{p}(l)) \exp(-\int_{0}^{l}\sigma(\mathbf{p}(l'))dl') dl.
\end{equation}
To simplify the notation without losing generality, we rewrite the Eqn.~\ref{equation-volume-rendering-color} as:
\begin{equation}
\label{equation-volume-rendering-color-sim}
    \mathbf{C} = \int_{0}^{\infty} \mathbf{c}(l) \sigma(l) \exp(-\int_{0}^{l}\sigma(l')dl') dl.
\end{equation}

\citet{volume_splatting} propose to choose each of the kernels which compose the density field to have finite intersection intervals on the ray and assume that there is no overlapping between intersection intervals of any two of them. These kernels can then be sorted based on their intersection interval along the ray. By plugging the Eqn.~\ref{equation-decomposition} into the Eqn.~\ref{equation-volume-rendering-color-sim}, we get:
\begin{equation}
\label{equation-volume-rendering-decom}
\begin{aligned}
    \mathbf{C} =& \int_{0}^{\infty} \sum_{i=1}^{n} \mathbf{c}(l) \omega_{i}\mathcal{K}_i(l) \exp(-\int_{0}^{l}\omega_{i}\mathcal{K}_i(l')dl') \\
    & \prod_{j=1}^{i-1} \exp(-\int_{0}^{l}\omega_{j}\mathcal{K}_j(l')dl') dl.
\end{aligned}
\end{equation}

Relying on the assumption that each kernel has finite intersection interval and there is no overlapping, we can further rewrite the Eqn.~\ref{equation-volume-rendering-decom} as:
\begin{equation}
\label{equation-volume-rendering-decom-sim}
\begin{aligned}
    \mathbf{C} =& \sum_{i=1}^{n} (\prod_{j=1}^{i-1} \exp(-\int_{0}^{\infty}\omega_{j}\mathcal{K}_j(l')dl')) \\
    & (\int_{0}^{\infty} \mathbf{c}(l) \omega_{i}\mathcal{K}_i(l) \exp(-\int_{0}^{l}\omega_{i}\mathcal{K}_i(l')dl') dl).
\end{aligned}
\end{equation}

\citet{volume_splatting} then assume the color is constant within the intersection interval of each individual kernel and ignores the self-occlusion. Eqn.~\ref{equation-volume-rendering-decom-sim} can then be written as:
\begin{equation}
    \mathbf{C} \approx \sum_{i=1}^{n} (\prod_{j=1}^{i-1} \exp(-\int_{0}^{\infty}\omega_{j}\mathcal{K}_j(l')dl')) (\mathbf{c}_i \int_{0}^{\infty} \omega_{i}\mathcal{K}_i(l) dl), 
\end{equation}
where $\mathbf{c}_i$ denotes the constant color value within the intersection interval of $i^\text{th}$ kernel.

\citet{volume_splatting} then propose to expand the exponential term using the Taylor series and defines a value $\rho_i = \omega_i \int_{0}^{\infty} \mathcal{K}_{i}(l)dl$, which is called the footprint function with respect to the current ray. We can then reach:
\begin{equation}
\label{equation-approx-final}
\begin{aligned}
    \mathbf{C} &\approx \sum_{i=1}^{n} (\prod_{j=1}^{i-1} (1 - \rho_{j})) (\mathbf{c}_i \rho_{i}) \\
    &= \sum_{i=1}^{n} \mathbf{c}_i \rho_{i} (\prod_{j=1}^{i-1} (1 - \rho_{j})), 
\end{aligned}
\end{equation}
which corresponds to the Eqn. 2 in the main paper. 
Notice that, even though the kernels are used to decompose the density field, due to the Taylor series expansion, $\rho_{i}, i=1,2,...,n$ has to be within the range $[0, 1]$. %
Also, notice that, in a more general sense, the footprint function is the integration of the density along the intersection interval, but it is assumed that there is only one kernel constituting the density there, therefore, the footprint function is simplified as the integration of the kernel values.

Therefore, as long as we can evaluate the footprint function $\rho$ easily and even differentiably, we can then reach an algorithm for efficient volume rendering, and even inverse rendering \cite{3dgs} under certain approximations.

As a summary, we identify following factors that make the volume splatting an approximate rendering algorithm:
\begin{itemize}
    \item Self-occlusion is ignored.
    \item Transmittance term is approximated through the Taylor expansion.
    \item It is assumed that there is no overlapping between the intersection intervals of any two kernels, but in practice, it is not the case, which leads to the sorting not clearly defined, as observed in \cite{radl2024stopthepop}.
\end{itemize}

A common choice of kernel is to use 3D Gaussian kernel as in \cite{volume_splatting, 3dgs}. In \cite{3dgs}, the footprint function of 3D Gaussian kernel is further approximated due to the perspective transform, which introduces additional bias.

A recent emerging interest \cite{2dgs, gaussian_surfel} is to use 2D Gaussian kernel, which is then known as Gaussian surfel \cite{Surfels}. 
In this case, perspective projection is no longer a problem because 
{the intersection between the 2D Gaussian kernel and the ray can be efficiently calculated \cite{2dgs}.}
And the sorting is always well-defined because the intersection interval along the ray is in general a point instead of an interval in the case of 3D Gaussian. However, the footprint function is then undefined because the integrand is a discontinuous function which has a finite value at the place where ray intersects with the 2D Gaussian kernel, and zero otherwise. The workaround proposed in \cite{2dgs} is to simply use the finite value as the footprint function, which still introduces further bias.

A rigorous exact rendering algorithm with efficient splatting is still an open problem, and we manage to solve it.

\subsubsection{Derivation of Refined Splatting Algorithm}
To move towards the exact rendering, we first address two approximations in the aforementioned derivation. We do not ignore the self-occlusion and do not expand the transmittance term. This corresponds to the Eqn. 7 in the main paper.

Specifically, from the Eqn.~\ref{equation-volume-rendering-decom-sim}, we have:
\begin{equation}
\label{equation-refined-color}
\begin{aligned}
    \mathbf{C} =& \sum_{i=1}^{n} (\prod_{j=1}^{i-1} \exp(-\rho_{j})) (\mathbf{c}_i \int_{0}^{\infty} \omega_{i}\mathcal{K}_i(l) \\
    &\exp(-\int_{0}^{l}\omega_{i}\mathcal{K}_i(l')dl') dl).
\end{aligned}
\end{equation}

Notice that:
\begin{equation}
    \frac{d}{dl}\exp(-\int_{0}^{l}\omega_{i}\mathcal{K}_i(l')dl') = -\omega_{i}\mathcal{K}_i(l)\exp(-\int_{0}^{l}\omega_{i}\mathcal{K}_i(l')dl')
\end{equation}

Therefore, from Eqn.~\ref{equation-refined-color} we have:
\begin{equation}
\label{refined-volume-rendering}
\begin{aligned}
    \mathbf{C} &= \sum_{i=1}^{n} (\prod_{j=1}^{i-1} \exp(-\rho_{j})) (\mathbf{c}_i (-\exp(-\int_{0}^{l}\omega_{i}\mathcal{K}_i(l')dl'))|^{\infty}_{0}) \\
    &=\sum_{i=1}^{n} (\prod_{j=1}^{i-1} \exp(-\rho_{j})) (\mathbf{c}_i (1-\exp(-\int_{0}^{\infty}\omega_{i}\mathcal{K}_i(l')dl'))) \\
    &=\sum_{i=1}^{n} (\prod_{j=1}^{i-1} \exp(-\rho_{j})) (\mathbf{c}_i (1-\exp(-\rho_{i}))) \\
    &=\sum_{i=1}^{n} \mathbf{c}_i (1-\exp(-\rho_{i})) \prod_{j=1}^{i-1} \exp(-\rho_{j}), 
\end{aligned}
\end{equation}
which corresponds to the Eqn. 7 in the main paper. 
Therefore, there is \textbf{no restriction on the value range} of $\rho_i, i=1,2,...,n$.

\subsection{Parameterize and Render Geometry Field with Gaussian Surfels.}
We focus on deriving the Eqn. 15 in the paper here. 
Recall that our goal is to evaluate the following equation:
\begin{equation}
\label{eqn-eqn-15}
    \rho_i = \int_{t_i - h/\cos\theta_i}^{t_i + h/\cos\theta_i} \frac{\psi(-F(\mathbf{x}(t)))}{\Psi(-F(\mathbf{x}(t)))} ||\nabla F(\mathbf{x}(t))||\cdot|\boldsymbol{\omega}\cdot\mathbf{n}(\mathbf{x}(t))| dt, 
\end{equation}
where $\Psi(\cdot)$ is the CDF of the standard normal distribution and $\psi(\cdot) = \Psi'(\cdot)$.
Besides, when $|t - t_i| < h/\cos\theta_i$, 
\begin{equation}
\begin{aligned}
    F(\mathbf{x}(t)) &= f_i \times \left(1-\frac{|t-t_i|}{h/\cos\theta_i}\right) -c, \\
    ||\nabla F(\mathbf{x}(t)) || &= f_i / h, \\
    \cos\theta_i &= |\boldsymbol{\omega}\cdot\mathbf{n}_i|. \\
\end{aligned}
\end{equation}
Recall that $\nabla F(\mathbf{x}(t))$ is parallel to $\mathbf{n}_i$ and we have $\mathbf{n}(\mathbf{x}(t))=\nabla F(\mathbf{x}(t)) / ||\nabla F(\mathbf{x}(t))||$, which implies that $\cos\theta_i = |\boldsymbol{\omega}\cdot\mathbf{n}(\mathbf{x}(t))|$.

Notice that \cref{eqn-eqn-15} is a symmetric integration with respect to $t_i$, therefore, 
\begin{equation}
\label{eqn-eqn-15-sim-1}
\begin{aligned}
    \rho_i &= 2\int_{t_i}^{t_i + h/\cos\theta_i} \frac{\psi(-F(\mathbf{x}(t)))}{\Psi(-F(\mathbf{x}(t)))} ||\nabla F(\mathbf{x}(t))||\cdot|\boldsymbol{\omega}\cdot\mathbf{n}(\mathbf{x}(t))| dt \\
    &= 2\int_{t_i}^{t_i + h/\cos\theta_i} \frac{\psi(-F(\mathbf{x}(t)))}{\Psi(-F(\mathbf{x}(t)))} \cdot(f_i/h)\cdot(\cos\theta_i) dt.
\end{aligned}
\end{equation}

Notice that, when $t\in[t_i, t_i + h/\cos\theta_i]$:
\begin{equation}
\begin{aligned}
    \frac{d}{dt}\ln\Psi(-F(\mathbf{x}(t))) &= \frac{\psi(-F(\mathbf{x}(t)))}{\Psi(-F(\mathbf{x}(t)))}\cdot (-F(\mathbf{x}(t)))' \\
    &= \frac{\psi(-F(\mathbf{x}(t)))}{\Psi(-F(\mathbf{x}(t)))}\cdot (c-f_i(1 - \frac{t-t_i}{h/\cos\theta_i}))' \\
    &= \frac{\psi(-F(\mathbf{x}(t)))}{\Psi(-F(\mathbf{x}(t)))}\cdot \frac{f_i}{h/\cos\theta_i} \\
    &= \frac{\psi(-F(\mathbf{x}(t)))}{\Psi(-F(\mathbf{x}(t)))}\cdot \frac{f_i}{h}\cdot \cos\theta_i.
\end{aligned}
\end{equation}

Therefore, \cref{eqn-eqn-15-sim-1} is then:
\begin{equation}
    \begin{aligned}
        \rho_i &= 2\int_{t_i}^{t_i + h/\cos\theta_i} \frac{\psi(-F(\mathbf{x}(t)))}{\Psi(-F(\mathbf{x}(t)))} \cdot(f_i/h)\cdot(\cos\theta_i) dt \\
        &= 2 \ln\Psi(-F(\mathbf{x}(t)))|^{t_i + h/\cos\theta_i}_{t_i} \\
        &= 2 (\ln\Psi(-(-c)) - \ln\Psi(-(f_i - c))) \\
        &= 2 (\ln\Psi(c) - \ln\Psi(c - f_i)) \\
        &= -2\ln\Psi(c - f_i), 
    \end{aligned}
\end{equation}
where $c$ is a large positive number such that $\Psi(c)=1$, and therefore $\ln\Psi(c)=0$.

Notice that we enable the calculation of footprint function by using the extrusion with width $2h$, and it can be seen as equivalent to the original case without the extrusion, by letting $h\rightarrow0$.

\paragraph{Discussion about Overlapping.} It is less obvious that the intersections between Gaussian surfels and the ray could overlap, i.e., the intersection points coincide, but we argue that it is indeed the case in practice.

First of all, we explicitly utilize the depth distortion loss to promote the ray to intersect the visible surface exactly once. Additionally, assume we have a flattened surface in the space. It then requires certain number of Gaussian surfels to compose it. Since we explicitly enforce the depth-normal consistency to smooth the geometry, these Gaussian surfels could become coplanar and partially overlap with each other within the accuracy of floating point numbers. Notice that the partial overlapping here refers to the overlapping of Gaussian surfels in the 3D space instead of the overlapping of intersections with respect to the ray which we mainly talk about. Therefore, when a ray falls into the area where Gaussian surfels overlap, the intersections between Gaussian surfels and the ray then coincide.

\section{Additional Implementation Details}
\label{sec:imp-details}
As to the depth distortion loss, the regularization weight is set to $1000$ for the DTU dataset, $10$ for the BlendedMVS dataset, and $0$ for the Mip-NeRF 360 dataset.

\subsection{Geometry Field Splatting}
We base our rasterizer implementation on that of 2DGS \cite{2dgs} using CUDA. 
Given a ray, the footprint function for the $i^\text{th}$ Gaussian surfel is defined as:
\begin{equation}
\begin{aligned}
    \rho_i &= -2\ln\Psi(c - f_i) \\
    &= -2\ln(0.5+0.5\text{erf}(\frac{c - f_i}{\sqrt{2}})). 
\end{aligned}
\end{equation}
We choose $c=3$, because $\Psi(3)\approx0.999$. Therefore, 
\begin{equation}
\begin{aligned}
    \rho_i = -2\ln(0.5+0.5\text{erf}(\frac{3 - f_i}{\sqrt{2}})), 
\end{aligned}
\end{equation}
We further clamp the maximum of $f_i$ to be $4.28$ such that the converted opacity $1-\exp(-\rho_i)\approx0.99$, because as in the original 3DGS \cite{3dgs} and 2DGS \cite{2dgs} implementations, the maximum of opacity is clamped to be $0.99$. Namely, 
\begin{equation}
\label{eqn-rho}
\begin{aligned}
    \rho_i = -2\ln(0.5+0.5\text{erf}(\frac{3 - \min\{f_i, 4.28\}}{\sqrt{2}})). 
\end{aligned}
\end{equation}
However, we find that when $f_i$ is small, the derivative of $\rho_i$ is close to zero and propagates almost zero gradients back to the parameter due to the numerical error, which hinders the optimization.

To alleviate the numerical error, we approximate \cref{eqn-rho} with a polynomial function:
\begin{equation}
    \rho_i \approx 0.03279 (\min\{f_i, 4.28\})^{3.4}, 
\end{equation}

We plot functions $F(x) = -2\ln(0.5 + 0.5\text{erf}(\frac{3-x}{\sqrt{2}}))$ and $G(x)=0.03279x^{3.4}$ in \cref{fig:functions}, and it can be seen that these two functions are close. %

\subsection{Remedy Loss Landscape Defects}
The color propagation per ray is approximated by propagating color in the $\mathbb{R}^3$ space. However, propagating the colors of all Gaussian surfels into the color of every Gaussian surfel is impractical.

Therefore, we use k-closest-point \cite{ravi2020pytorch3d} algorithm based on the centers of Gaussian surfels to identify $10$ closest Gaussian surfels to each Gaussian surfel every $100$ iterations. 
The color of every Gaussian surfel is then blended based on these $10$ closest Gaussian surfels which include itself. 
The choice of number of closest Gaussian surfels is made by ensuring the program does not cause out of memory error with either of our two color representations on a standard consumer-level graphics card with around $10$ GB memory, on the DTU dataset \cite{dtu}.

\subsection{Improve Color Representation}
Following the default setting of \cite{torch-ngp}, the shallow MLP is implemented with $2$ hidden layers, and the spherical harmonics encoding of directions has degree $4$.

\begin{figure}[t]
    \centering
    \includegraphics[width=\linewidth]{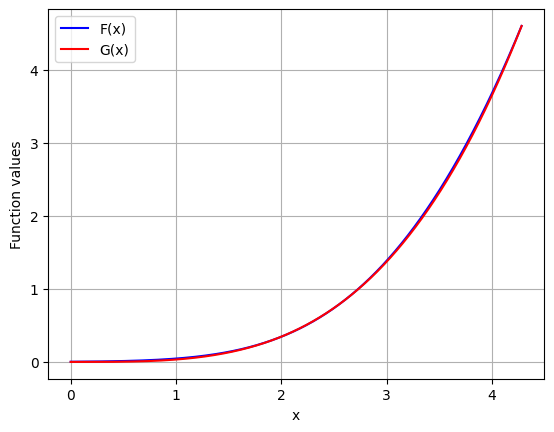}
    \caption{Plotting of function values of $F(x) = -2\ln(0.5 + 0.5\text{erf}(\frac{3-x}{\sqrt{2}}))$ and $G(x)=0.03279x^{3.4}$ over the range $[0, 4.28]$.}
    \label{fig:functions}
\end{figure}

\begin{table*}[]
\centering
\footnotesize
\setlength{\tabcolsep}{5pt}
\begin{tabular}{l|ccccccccccccccc|c}
\hline
\textbf{Methods}                       & \textbf{24} & \textbf{37} & \textbf{40} & \textbf{55} & \textbf{63} & \textbf{65} & \textbf{69} & \textbf{83} & \textbf{97} & \textbf{105} & \textbf{106} & \textbf{110} & \textbf{114} & \textbf{118} & \textbf{122} & \textbf{Avg. $\downarrow$} \\ \hline
Ours (SH) & 
\cellcolor{red!25}0.38 & 
\cellcolor{red!25}0.63 & 
\cellcolor{red!25}0.30 & 
\cellcolor{orange!25}0.35 & 
\cellcolor{orange!25}0.79 & 
\cellcolor{orange!25}0.63 & 
\cellcolor{orange!25}0.65 & 
\cellcolor{red!25}1.10 & 
1.21 & 
\cellcolor{red!25}0.62 & 
\cellcolor{orange!25}0.48 & 
\cellcolor{yellow!25}1.18 & 
\cellcolor{orange!25}0.33 & 
\cellcolor{red!25}0.42 & 
\cellcolor{orange!25}0.38 & 
\cellcolor{orange!25}0.63
\\
- Geometry Field Splatting &
0.50 &
\cellcolor{yellow!25}0.67 &
\cellcolor{yellow!25}0.35 &
0.43 &
0.94 &
0.89 &
\cellcolor{yellow!25}0.70 &
\cellcolor{yellow!25}1.17 &
\cellcolor{red!25}1.18 &
\cellcolor{yellow!25}0.65 &
0.57 &
\cellcolor{red!25}1.02 &
0.39 &
0.61 &
0.48 &
0.70
\\
- Remedy Loss Landscape Defects & 
\cellcolor{yellow!25}0.42 & 
\cellcolor{yellow!25}0.67 & 
\cellcolor{yellow!25}0.35 & 
\cellcolor{yellow!25}0.36 & 
\cellcolor{yellow!25}0.85 & 
\cellcolor{yellow!25}0.69 & 
\cellcolor{yellow!25}0.70 & 
1.19 & 
\cellcolor{orange!25}1.19 & 
\cellcolor{yellow!25}0.65 & 
\cellcolor{yellow!25}0.49 & 
1.24 & 
\cellcolor{yellow!25}0.36 & 
\cellcolor{yellow!25}0.46 & 
\cellcolor{yellow!25}0.40 & 
\cellcolor{yellow!25}0.67
\\
+ Per-Ray Sorting & 
\cellcolor{orange!25}0.39 & 
\cellcolor{orange!25}0.65 & 
\cellcolor{orange!25}0.31 & 
\cellcolor{red!25}0.33 & 
\cellcolor{red!25}0.78 & 
\cellcolor{red!25}0.49 & 
\cellcolor{red!25}0.61 & 
\cellcolor{orange!25}1.14 & 
\cellcolor{yellow!25}1.20 & 
\cellcolor{orange!25}0.64 & 
\cellcolor{red!25}0.46 & 
\cellcolor{orange!25}1.07 & 
\cellcolor{red!25}0.29 & 
\cellcolor{orange!25}0.44 & 
\cellcolor{red!25}0.37 & 
\cellcolor{red!25}0.61
\\ \hline
Ours (Latent) & 
\cellcolor{yellow!25}0.40 & 
\cellcolor{orange!25}0.59 & 
\cellcolor{orange!25}0.39 & 
0.38 & 
\cellcolor{red!25}0.72 & 
\cellcolor{orange!25}0.59 & 
\cellcolor{yellow!25}0.65 & 
\cellcolor{red!25}1.08 & 
\cellcolor{red!25}0.93 & 
\cellcolor{red!25}0.59 & 
\cellcolor{orange!25}0.50 & 
\cellcolor{yellow!25}0.67 & 
\cellcolor{orange!25}0.34 & 
\cellcolor{red!25}0.47 & 
\cellcolor{yellow!25}0.40 & 
\cellcolor{orange!25}0.58
\\
- Geometry Field Splatting &
0.52 &
0.79 &
0.47 &
0.52 &
0.91 &
0.71 &
0.84 &
1.17 &
1.07 &
0.63 &
0.65 &
0.82 &
0.42 &
0.75 &
0.53 &
0.72
\\
- Remedy Loss Landscape Defects &
\cellcolor{red!25}0.38 &
\cellcolor{red!25}0.56 &
\cellcolor{orange!25}0.39 &
\cellcolor{orange!25}0.36 &
0.80 &
\cellcolor{yellow!25}0.60 &
0.68 &
\cellcolor{yellow!25}1.13 &
\cellcolor{orange!25}0.95 &
0.64 &
\cellcolor{orange!25}0.50 &
\cellcolor{orange!25}0.65 &
\cellcolor{yellow!25}0.35 &
\cellcolor{orange!25}0.48 &
\cellcolor{yellow!25}0.40 &
\cellcolor{yellow!25}0.59
\\
+ Per-Ray Sorting & 
0.43 & 
0.64 & 
\cellcolor{red!25}0.38 & 
\cellcolor{red!25}0.34 & 
\cellcolor{orange!25}0.73 & 
\cellcolor{red!25}0.50 & 
\cellcolor{red!25}0.57 & 
\cellcolor{yellow!25}1.13 & 
1.00 & 
0.65 & 
\cellcolor{red!25}0.49 & 
\cellcolor{red!25}0.59 & 
\cellcolor{red!25}0.30 & 
\cellcolor{orange!25}0.48 & 
\cellcolor{red!25}0.38 & 
\cellcolor{red!25}0.57
\\
- Ray Direction Conditioning & 
0.47 & 
\cellcolor{yellow!25}0.61 & 
0.42 & 
0.39 & 
0.81 & 
0.78 & 
0.75 & 
\cellcolor{yellow!25}1.13 & 
\cellcolor{yellow!25}0.98 & 
\cellcolor{orange!25}0.60 & 
\cellcolor{yellow!25}0.56 & 
\cellcolor{yellow!25}0.67 & 
0.39 & 
\cellcolor{yellow!25}0.55 & 
0.44 & 
0.64
\\
- Reflected Ray Direction Conditioning & 
\cellcolor{orange!25}0.39 & 
\cellcolor{yellow!25}0.61 & 
\cellcolor{yellow!25}0.40 & 
\cellcolor{yellow!25}0.37 & 
\cellcolor{yellow!25}0.78 & 
\cellcolor{yellow!25}0.60 & 
\cellcolor{orange!25}0.62 & 
\cellcolor{orange!25}1.09 & 
\cellcolor{yellow!25}0.98 & 
\cellcolor{yellow!25}0.61 & 
\cellcolor{orange!25}0.50 & 
0.80 & 
\cellcolor{yellow!25}0.35 & 
\cellcolor{red!25}0.47 & 
\cellcolor{orange!25}0.39 & 
0.60
\\ \hline
\end{tabular}
\caption{Quantitative evalution on the DTU dataset based on the Chamfer Distance for different ablation models. The best metric is highlighted in red, the second best metric is highlighted in orange, and the third best metric is highlighted in yellow.}
\label{table:ablation}
\end{table*}

\section{Additional Evaluation Results}
\label{sec:eval-results}

\begin{figure}[t]
    \centering
    \includegraphics[width=\linewidth]{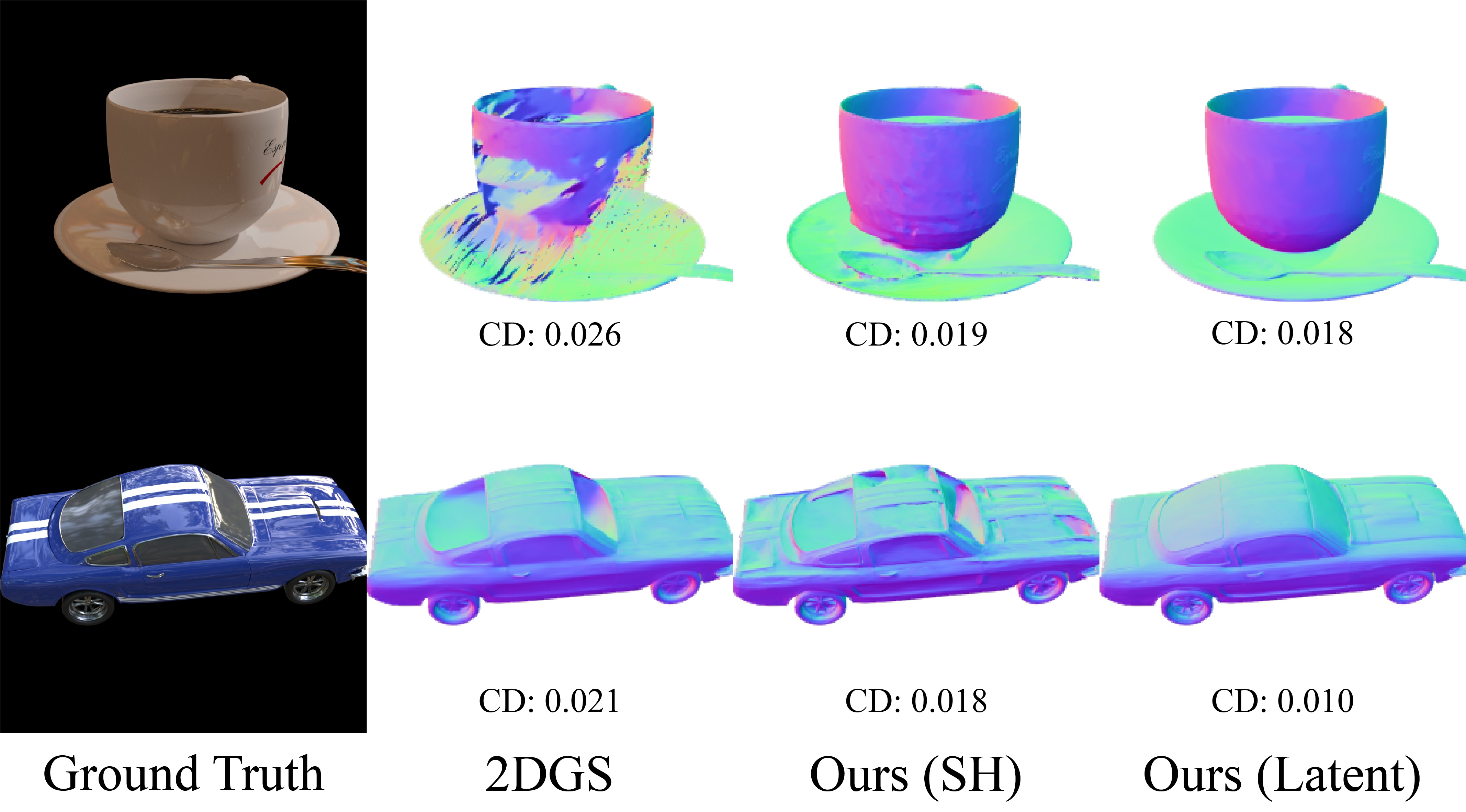}
    \caption{{Comparisons of normal maps and chamfer distances (denoted as ``CD''; lower is better) of reconstructed geometry for highly specular surfaces (first column) from the ShinyBlender dataset \cite{refnerf}.}}
    \label{fig:ablation_color_representation}
\end{figure}

\paragraph{Evaluation Details.} 
We find that the evaluation of mesh quality is also dependent on how well the mesh is cleaned after extraction. 
Typically, SDF-based approaches use the marching cube algorithm to produce closed-surface meshes which differ from the open-surface ground-truth and then contain redundant parts. 
In contrast, splatting-based approaches typically use the TSDF fusion algorithm to produce open-surface meshes based on the depth maps, which are more concise. 
It then becomes necessary to clean triangles, which are invisible from all training views, to have a fair unified evaluation protocol. In practice, we combine the scripts from \cite{long2022sparseneus} and \cite{2dgs} to clean the mesh before evaluation.
As to the Neuralangelo, we find that its extracted mesh sometimes is enclosed in a sphere, which makes the mesh cleaning fail. Therefore, we manually remove the enclosing spheres if they exist for the extracted meshes of Neuralangelo before passing them into the evaluation.

\subsection{Quantitative Evaluation}
We evaluate the view synthesis quality of our method on the Mip-NeRF 360 dataset \cite{mipnerf360}, while comparing with NeRF \cite{nerf2020}, INGP \cite{instant-ngp}, MERF \cite{reiser2023merf}, BakedSDF \cite{BakedSDF}, MipNeRF 360 \cite{mipnerf360}, BOG \cite{BOG}, 3DGS \cite{3dgs}, SuGaR \cite{guedon2023sugar}, MipSplatting \cite{Yu2024MipSplatting}, 2DGS \cite{2dgs}, GOF \cite{GOF}, and RaDe-GS \cite{radegs}. We find that with geometry clearly defined and its corresponding regularization, the view synthesis is harmed.

As to the rendering speed during inference, since our proposed color propagation is irrelevant to the view, it can be achieved and cached before rendering, thus does not impact the rendering speed.

\begin{table}[t]
\centering
\footnotesize
\setlength{\tabcolsep}{3pt}
\begin{tabular}{c|ccc|ccc}
\hline
\textbf{}        & \multicolumn{3}{c|}{\textbf{Outdoor Scenes}}   & \multicolumn{3}{c}{\textbf{Indoor Scenes}}     \\
\textbf{Methods} & \textbf{PSNR $\uparrow$} & \textbf{SSIM $\uparrow$} & \textbf{LPIPS $\downarrow$} & \textbf{PSNR $\uparrow$} & \textbf{SSIM $\uparrow$} & \textbf{LPIPS $\downarrow$} \\ \hline
NeRF             & 21.46         & 0.458         & 0.515          & 26.84         & 0.790         & 0.370          \\
INGP      & 22.90         & 0.566         & 0.371          & 29.15         & 0.880         & 0.216          \\
MERF             & 23.19         & 0.616         & 0.343          & 27.80         & 0.855         & 0.271          \\
BakedSDF         & 22.47         & 0.585         & 0.349          & 27.06         & 0.836         & 0.258          \\
MipNeRF 360      & 24.47         & 0.691         & 0.283          & \cellcolor{red!25}{31.72}         & 0.917         & 0.180          \\ \hline
BOG              & 23.94         & 0.680         & 0.263          & 27.71         & 0.873         & 0.227          \\
3DGS             & 24.64         & {0.731}         & 0.234          & 30.41         & 0.920         & \cellcolor{yellow!25}{0.189}          \\
SuGaR            & 22.93         & 0.629         & 0.356          & 29.43         & 0.906         & 0.225          \\
MipSplatting    & \cellcolor{yellow!25}{24.65}         & 0.729         & 0.245          & \cellcolor{orange!25}{30.90}         & \cellcolor{yellow!25}{0.921}         & 0.194          \\
2DGS  & 24.18 & 0.703 & 0.287 & 30.06 & 0.909 & 0.213 \\
GOF              & \cellcolor{orange!25}{24.82}         & \cellcolor{orange!25}{0.750}         & \cellcolor{orange!25}{0.202}          & \cellcolor{yellow!25}{30.79}         & \cellcolor{orange!25}{0.924}         & \cellcolor{orange!25}{0.184}          \\
RaDe-GS          & \cellcolor{red!25}{25.17}         & \cellcolor{red!25}{0.764}         & \cellcolor{red!25}{0.199}          & 30.74         & \cellcolor{red!25}{0.928}         & \cellcolor{red!25}{0.165}          \\ \hline
Ours (SH)        & 24.40  &  \cellcolor{yellow!25}{0.734} & \cellcolor{yellow!25}{0.224} & 29.93 & 0.916 & 0.194 \\
Ours (Latent)   & 23.76 & 0.693 & 0.293 & 29.92 & 0.906 & 0.219 \\
\hline
\end{tabular}
\caption{Quantitative comparison on the Mip-NeRF 360 dataset based on the view synthesis.}
\label{table:view-synthesis}
\end{table}

\subsection{Qualitative Evaluation}
We show complete rendering of all cases on the DTU dataset \cite{dtu} and BlendedMVS dataset \cite{yao2020blendedmvs} between our method and 2DGS \cite{2dgs} in \cref{fig:dtu-full}, \cref{fig:bmvs-full}, and \cref{fig:bmvs-scene-full}. 
In the 2DGS algorithm, using the median depth in the regularization leads to better quantitative results, while using the mean depth in the regularization produces smoother geometry and more visually pleasing results. 
In \cref{fig:dtu-full}, we compare with 2DGS trained with both settings. We capture more geometric details than those trained with mean depth in 2DGS, and are free of cracks and holes compared to those trained with median depth in 2DGS. 
We also show the extracted mesh on a few scenes of MipNeRF 360 dataset in \cref{fig:fig:mipnerf-scene} using our method and TSDF fusion.
{We further provide preliminary comparison results for two selected cases on the ShinyBlender dataset \cite{refnerf} in Fig.~\ref{fig:ablation_color_representation} to demonstrate the capability of our method to generalize to highly specular surfaces.}

We also test our method on the Tanks\&Temples dataset \cite{tanks} but this dataset does not provide ground-truth camera parameters, which requires an iterative-closest-point procedure to align the extracted mesh and ground-truth mesh. After manual inspection, we find that the metrics heavily depend on how well this off-the-shelf alignment algorithm aligns the extracted mesh and ground-truth mesh, and it actually fails on two scenes. Therefore, we choose to only show our qualitative results in \cref{fig:tnt-scene-full}.

\subsection{Ablation Study}
We provide the complete evaluation results on the DTU dataset for all our ablation models in \cref{table:ablation}. Specifically, we evaluate with our two color representations, i.e., SH representation and latent representation. As to the latent representation, we also ablate on the choice of conditioning for MLP.
We conduct all the ablation experiments with the same hyper-parameters with the baseline.

By comparing the baseline with the ablation model without geometry field splatting, it is clear that the geometry field splatting instead of the original approximate rendering formulation in 2DGS significantly boosts the performance. However, we also find that the geometry field splatting makes the SH representation more sensitive to the specular surfaces, which is reflected in the ``scan110''. 
By comparing the baseline with the ablation model without the remedy to loss landscape defects, we can see that our proposed color blending consistently improves the performance on almost every case, especially for the SH representation. 
The ablation model which implements per-ray sorting achieves better averaged performance, but does not always achieve the best performance on every case. We argue that it may be due to the fact that the hyper-parameters are tailored for the global sorting approximation, and not tuned for the per-ray sorting case. Implementing \cite{radl2024stopthepop} with our algorithm while tuning relevant hyper-parameters may further improve our performance.

As to the latent representation, without ray direction conditioning or reflected ray direction conditioning, the reconstructed geometry suffers.

\section{Additional Discussions}
In this paper, we focus on evaluating our proposed geometry field representation and remedy to loss landscape defects, and therefore do not incorporate many other methods which could further improve the geometry reconstruction. However, we identify a few potential ways here which may further improve the geometry reconstruction with our method.

\paragraph{Anti-aliasing.} As discussed in \cite{2dgs}, a 2D Gaussian is degenerated into a line in the screen space while being observed from a slanted viewpoint. Besides, as discussed in \cite{volume_splatting}, the output image has to be band-limited to avoid aliasing. We follow \cite{2dgs} to only apply an object-space low-pass filter \cite{10.5555/2386366.2386369}, which actually does not follow our rendering model. Similar to \cite{Yu2024MipSplatting}, it would be useful to design an anti-aliasing method which is tailored for our rendering model.
\paragraph{Appearance Embedding.} We assume that the scene is static and the images are captured in the same lighting condition. 
However, there could be slight variance of lighting conditions, and brightness of the images, etc., due to the capture. Therefore, it would be useful to have an appearance embedding per image as in \cite{pgsr, radegs, GOF}. 
\paragraph{Color Representation.} {We choose to use a latent representation for color that is as simple as possible to improve on the specular cases and emphasize more on our proposed geometric representation. In the future, it will be beneficial to further extend our method with better latent representations as in \cite{gaussianshader} to work on highly specular surfaces.}
\paragraph{Multi-view Stereo Constraint.} Concurrently, 
\citet{pgsr} shows that the multi-view stereo constraint could be helpful for the geometry reconstruction quality using the volume splatting representation. It would be useful to apply the multi-view constraint with our clearly defined geometry, which may further improve the reconstruction.

\paragraph{Incorporate Monocular Priors.} In the dense-views geometry reconstruction task, the captured images are assumed to be sufficient to recover the geometry. However, there may not be clear standard to decide whether the captured images are sufficient to define a unique solution, and, in practice, the captured images are usually not enough. 
As in \cite{gaussian_surfel}, it would be useful to incorporate modern monocular priors, such as the estimated monocular normal (e.g., \cite{zeng2024rgb}), to compensate for the deficiency of input images, especially for the scene-level cases.

\paragraph{Densification Strategies.} We observe that, due to our refined splatting algorithm, a Gaussian surfel could have a much larger fully opaque area. Even though this property benefits the geometry reconstruction overall, it makes the optimization and densification harder, which is also reflected in the affected view synthesis quality. 
Since the optimization is stochastic, a Gaussian surfel could have a  large fully opaque area. Besides, the enforced depth-normal consistency loss effectively squeezes the Gaussian surfel, which makes the large fully opaque area difficult to reduce. It then blocks the optimization towards the other Gaussian surfels it occludes. 
Even though we periodically reset the geometry value on the Gaussian surfel to make it less opaque, we find that its geometry value will still recover fast instead of shrinking its size.

We find that using the densification strategy proposed by \cite{ye2024absgs} leads to smaller Gaussian surfels in general without losing the quality, which alleviates such a problem. However, a different strategy which may fully solve this problem is welcome.

\paragraph{Mesh Extraction.} 
We follow \cite{2dgs} to use the TSDF fusion to extract the mesh. Even though it gives reasonable results, it is slow as it operates on the CPU and it also cannot handle thin structures such as the strips on the wheels of the bicycle well. In contrast, we also find the marching tetrahedra proposed by \cite{GOF} works sub-optimally on the object-centric cases.
With the advance of explicit primitives, i.e., Gaussian kernels, it would be useful to have a mesh extraction method that may directly convert primitives of different representations.

\clearpage

\begin{figure*}[t!]
    \centering
    \vfill %
    \includegraphics[width=0.73\linewidth]{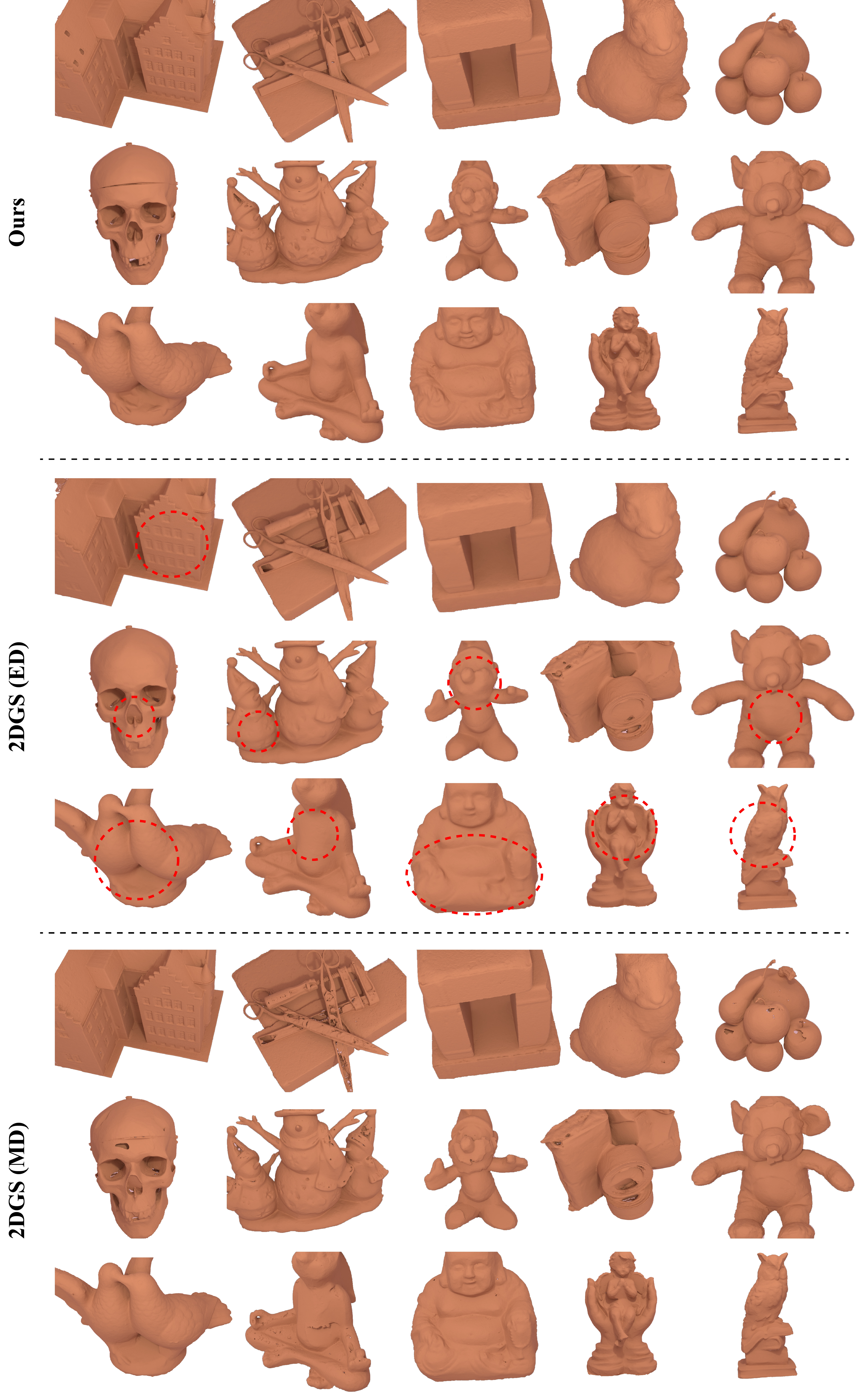}
    
    \caption{Comparison of rendering of meshes extracted by our method (at the top), 2DGS trained with mean depth in the regularization (in the middle), and 2DGS trained with median depth in the regularization (at the bottom) on the DTU dataset. Non-obvious differences are highlighted in red circles. The reader may wish to zoom into the electronic version in the figures.}
    \label{fig:dtu-full}
\end{figure*}

\clearpage

\begin{figure*}[t!]
    \centering
    \vfill %
    \includegraphics[width=\linewidth]{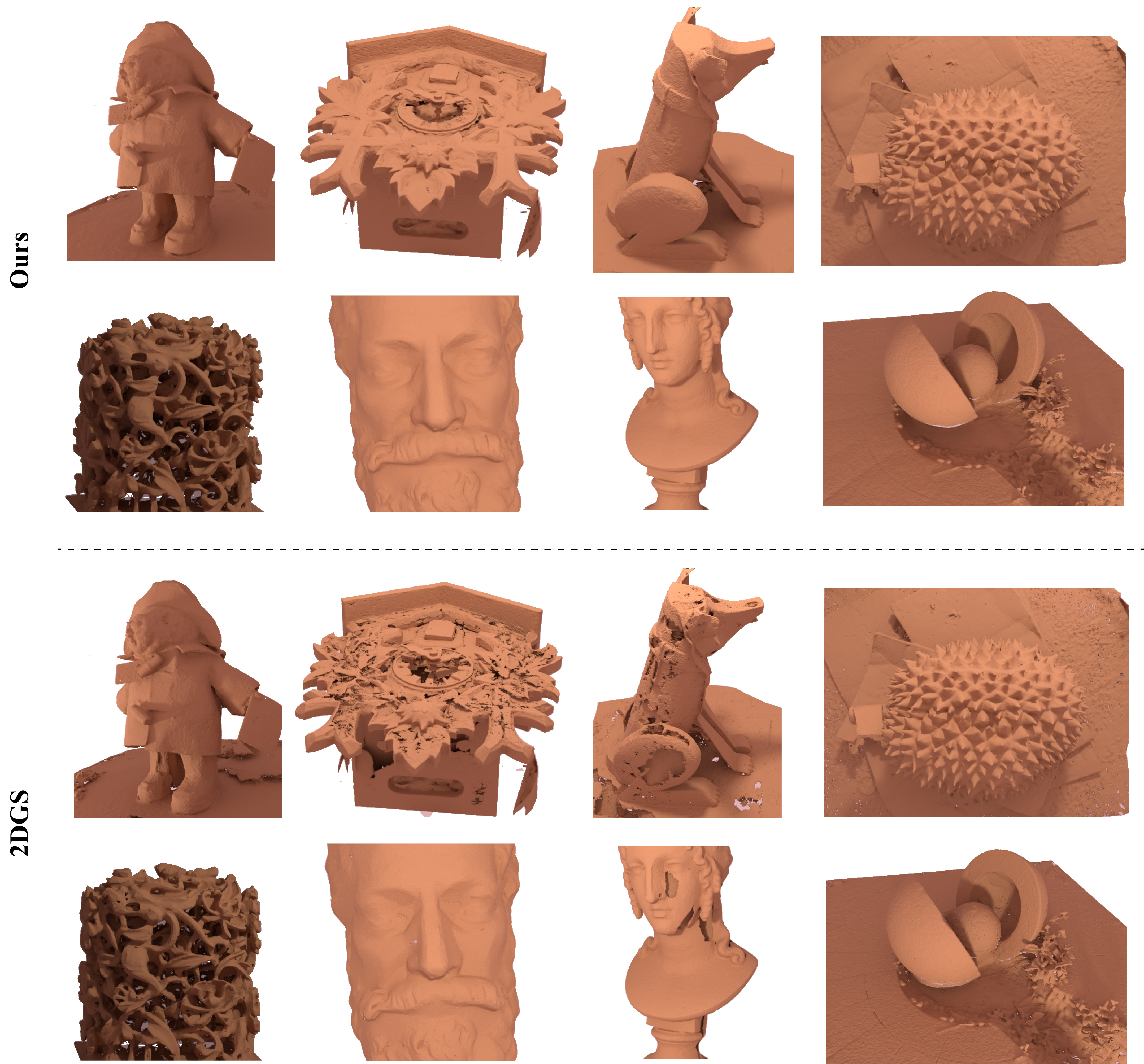}
    
    \caption{Comparison of rendering of meshes extracted by our method (at the top) and 2DGS (at the bottom) on the object cases of the BMVS dataset.}
    \label{fig:bmvs-full}
\end{figure*}

\clearpage

\begin{figure*}[t!]
    \centering
    \includegraphics[width=0.8\linewidth]{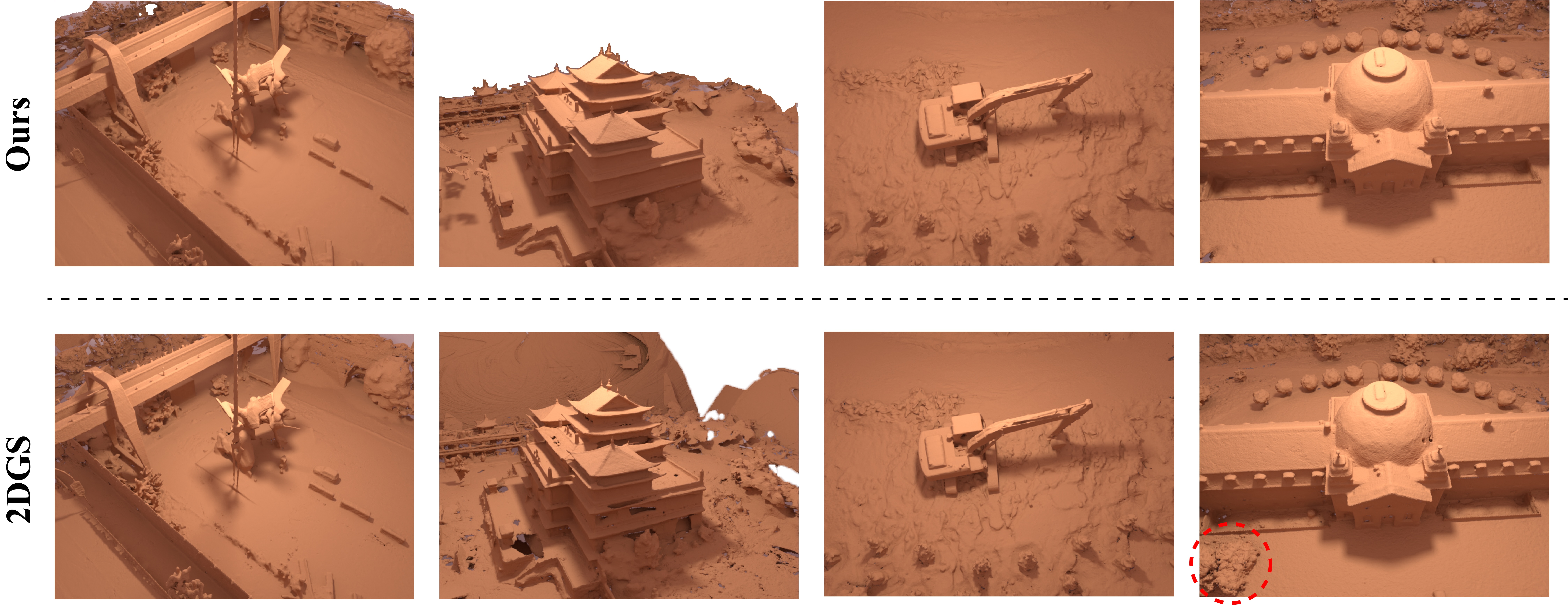}
    
    \caption{Comparison of rendering of meshes extracted by our method (at the top) and 2DGS (at the bottom) on the scene cases of BMVS dataset. Non-obvious differences are highlighted in red circles.}
    \label{fig:bmvs-scene-full}
\end{figure*}

\begin{figure*}[t!]
    \centering
    \includegraphics[width=0.85\linewidth]{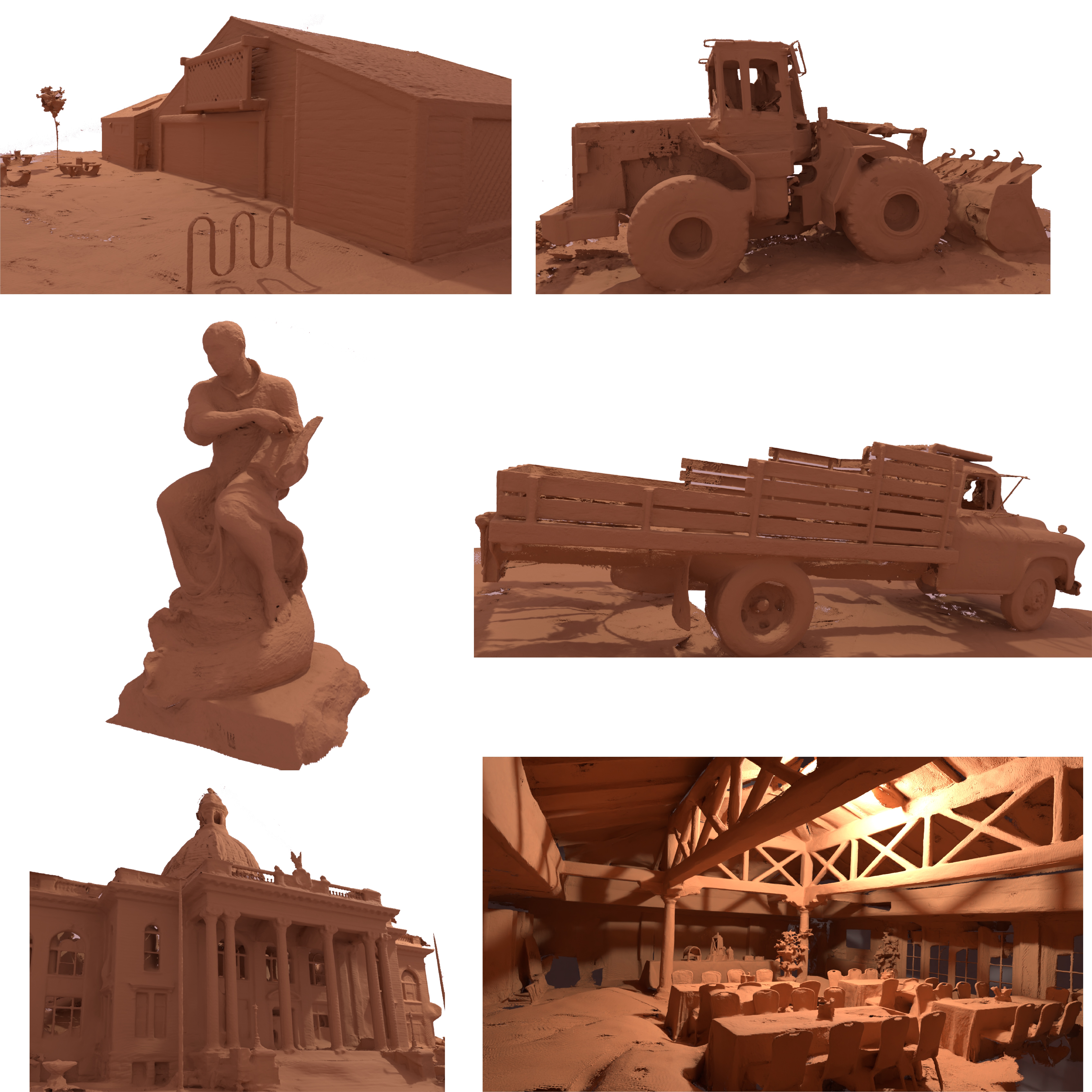}
    
    \caption{Rendering of meshes extracted by our method on the Tanks\&Temples dataset.}
    \label{fig:tnt-scene-full}
\end{figure*}

\clearpage

\begin{figure*}[t!]
    \centering
    \vfill %
    \includegraphics[width=0.8\linewidth]{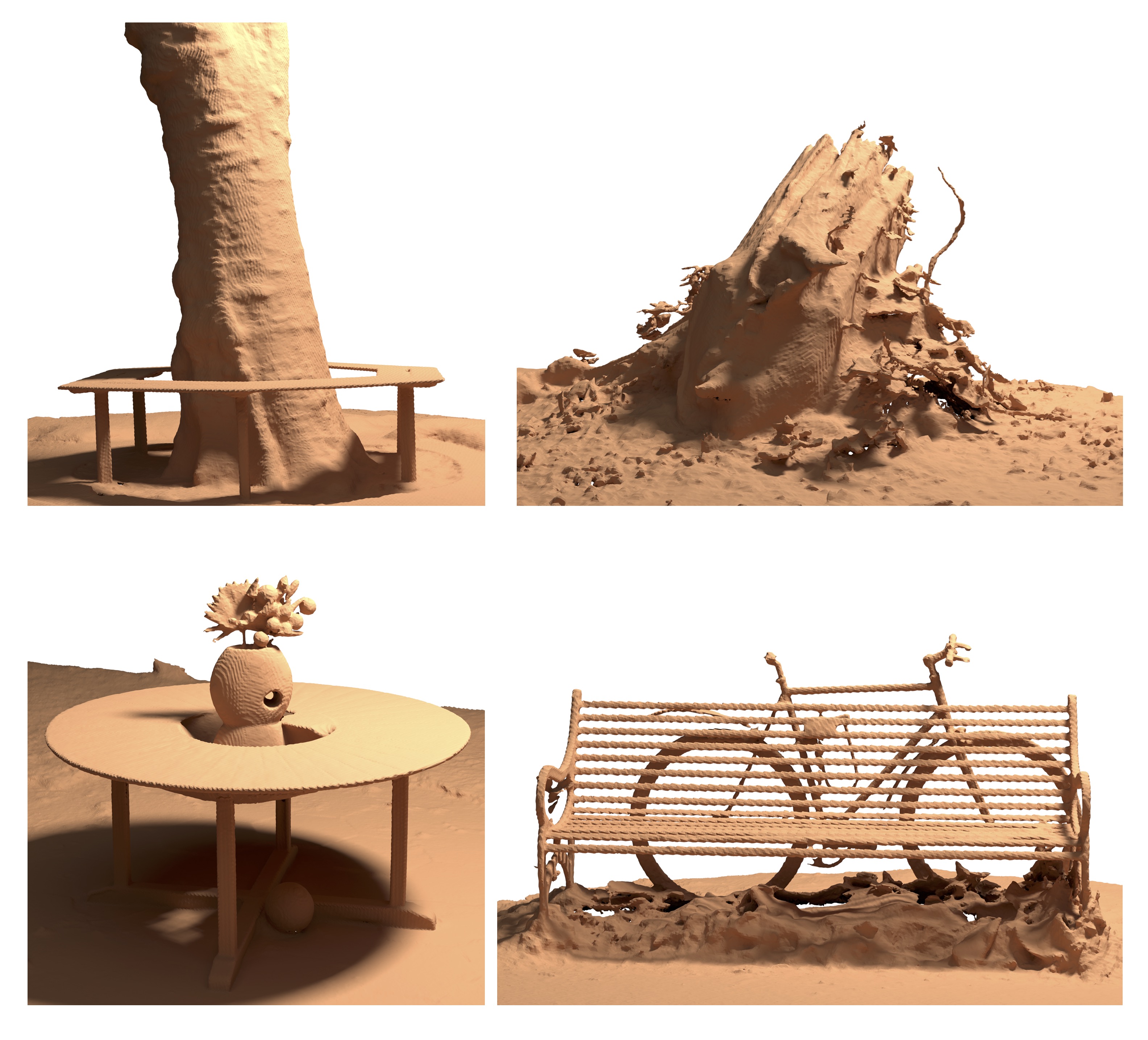}
    
    \caption{Rendering of meshes extracted by our method on the Treehill, Garden, Bicycle and Stump scenes of MipNeRF 360 dataset.}
    \label{fig:fig:mipnerf-scene}
\end{figure*} \fi

\end{document}